\RequirePackage{amsmath}
\documentclass[a4paper]{article}
\usepackage{amsthm}
\usepackage[a4paper, margin=3.18cm]{geometry}
\usepackage[T1]{fontenc}
\usepackage{amssymb}
\usepackage{graphicx}
\usepackage{hyperref}
\usepackage{color}
\usepackage{orcidlink}

\usepackage{enumitem}
\usepackage{units}
\usepackage{caption}
\usepackage{thm-restate}
\usepackage{subfig}
\usepackage{mathtools}
\usepackage{todonotes}

\newtheorem{theorem}{Theorem}
\newtheorem{lemma}[theorem]{Lemma}
\newtheorem{corollary}[theorem]{Corollary}
\newtheorem{observation}[theorem]{Observation}

\newcommand*{\Na}{\mathbb{N}}

\newcommand{\set}[2]{\{\, #1 \mid #2 \,\}}
\newcommand*{\bigO}{\mathcal{O}}
\newcommand*{\bigOmega}{\Omega}

\newcommand*{\cost}{\text{cost}}
\newcommand*{\s}{\mathbf{s}}

\date{}
\usepackage[backend=biber, style=numeric, doi=true, url=true, maxnames=99]{biblatex}
\usepackage{hyperref}
\usepackage{authblk}
\begin{document}

\title{Network Creation Games with 2-Neighborhood Maximization}

\renewcommand\Authfont{\normalsize}  
\renewcommand\Affilfont{\small}  
\renewcommand{\Authands}{, }  
\makeatletter
\renewcommand{\AB@authnote}[1]{}  
\makeatother

\author{Merlin de la Haye\textsuperscript{1}\orcidlink{0009-0000-9517-4609}, 
Pascal Lenzner\textsuperscript{2}\orcidlink{0000-0002-3010-1019}, 
Daniel Schmand\textsuperscript{3}\orcidlink{0000-0001-7776-3426}, and 
Nicole Schröder\textsuperscript{3}\orcidlink{0009-0008-2589-3809}}

\affil[1]{Hasso Plattner Institute, University of Potsdam, Potsdam, Germany\\ \texttt{merlin.delahaye@hpi.de}}
\affil[2]{Institute of Computer Science, University of Augsburg, Augsburg, Germany\\ \texttt{pascal.lenzner@uni-a.de}}
\affil[3]{University of Bremen, Bremen, Germany\\ \texttt{\{schmand,nic\_sch\}@uni-bremen.de}}

\maketitle

\begin{abstract}
\noindent Network creation games are well-established for investigating the decentralized formation of communication networks, like the Internet or social networks. In these games, selfish agents that correspond to network nodes strategically create costly edges to maximize their centrality in the formed network. We depart from this by focusing on the simpler objective of maximizing the 2-neighborhood. This seems natural for social networks, as an agent's connection benefit is typically provided by her neighbors and their neighbors but not by strangers further away. 

For this natural model, we study the existence, the structure and the quality both of Nash equilibria (NE) and greedy equilibria (GE). 
We give structural results on the existence of degree-2 paths and cycles, and we provide tight constant bounds on the diameter. In contrast to most previous network creation game research, our bounds on the diameter are independent of edge cost $\alpha$ and the number of agents $n$. Also, bounding the diameter does not imply bounding the price of anarchy, which calls for other methods. Using them, we obtain non-trivial bounds on the price of anarchy, including a $\Omega\left(\log\left(\frac{n}{\alpha}\right)\right)$ lower bound for NE, and a tight linear bound for~GE for~low~$\alpha$.
\smallskip

\end{abstract}

\section{Introduction} \label{sec:Intro}
Networks are essential in many domains, ranging from the shareholder value of multi-national corporations to career prospects of individuals. Their core feature is that they allow individual firms or people to profit from the abilities, knowledge, products, or results of their network neighbors. However, the benefit of network neighbors is not restricted to those neighbors, but for many applications, also the neighbors' neighbors, or the friends of our friends, are crucial. This holds true for business interactions, job or housing search and even for dating. In all these settings, a common acquaintance, friend, or business partner can serve as a broker for valuable information. This explains why real-world networks often feature high-degree nodes~\cite{nr-aaai15,Bar16}, i.e., very well-connected individuals.
Thus, from the perspective of an individual in such a network, it is crucial to connect to the right set of other individuals to maximize the personal benefit obtained from the network~\cite{burt1995structural,KleinbergSTW08}. This benefit consists of direct benefit from the network neighbors but also indirect benefit from the neighbors of those neighbors. Interestingly, typically this benefit does not transcend much further, i.e., to friends of friends of friends, since for such interactions a joint neighbor that can serve as a trusted information broker is lacking and the communication overhead increases drastically. There is thus a huge qualitative difference between interacting with friends or friends of friends compared to interacting with random strangers.

Another key observation is that many important networks, like the Internet and all kinds of social networks, ranging from business partner networks over friendship networks to online social networks, are not centrally designed or operated. Instead, they emerge from the strategic connection decisions of the individual network participants. Thus, for studying networks, an agent-based game-theoretic approach is natural~\cite{Papadimitriou01}.

In this paper, we study a simple and elegant model for the formation of networks that features agents aiming for a trade-off between the costly establishment and maintenance of network links and the derived direct and indirect benefit of those links. This is measured with the size of their $2$-neighborhood in the created network, i.e.\ the number of network nodes in hop-distance at most~$2$. Our model arises as a natural variant of Jackson and Wolinsky's seminal connections model~\cite{JW96}, where the mutual benefit of individuals in a network decreases with their hop-distance and it can also be seen as a modification of the well-known network creation game by Fabrikant, Luthra, Maneva, Papadimitriou, and Shenker~\cite{fabrikantNetwork2003}, where instead of minimizing the distance to all network nodes, maximizing the size of the 2-neighborhood is used as a simplified proxy. 

Our model is a special case of the star celebrity game setting introduced by {\`A}lvarez, Blesa, Duch, Messegu{\'e}, and Serna~\cite{ABDMS16celebrity}. We extend their work by showing that even the considered special case admits a structurally rich set of Nash equilibria. Also, we establish additional structural properties and we improve the diameter bound. Moreover, we show that the game does not have the finite improvement property and we also consider greedy equilibria as solution concept. Finally, since for our setting the tight price of anarchy bound presented in~\cite{ABDMS16celebrity} crucially relies on inconsistent tie-breaking, we provide a new non-trivial lower bound that also holds for consistent tie-breaking.

\subsection{Model and Preliminaries}

Let $G = (V,E)$ be any undirected graph, where $V$ is the set of nodes and $E$ is the set of edges. Let $d_G(u,v)$ denote the hop-distance between nodes $u,v\in V$, i.e., the number of edges on a shortest $u$-$v$-path in $G$.
For any node $u\in V$, let $N_2^G(u) = \{v\in V\mid d_G(u,v)\leq 2\}$ denote the \emph{2-neighborhood} of node $u$ in $G$, i.e., the set of nodes in $G$ that are in hop-distance at most $2$ from~$u$. 

We consider a strategic game, called the \emph{2-neighborhood maximization game (2-NMG)}, where $n$ agents $V = \{1,\dots,n\}$ that correspond to nodes in a graph can strategically create costly edges to other nodes to maximize their 2-neighborhood in the arising graph. Formally, every agent $u\in V$ selects a \emph{strategy} $S_u \subseteq V\setminus ~\{u\}$, which is any subset of the other agents. Let $\s = (S_1,\dots,S_n)$ denote the corresponding \emph{strategy profile}. For convenience, we use $\s = (S_i,\s_{-i})$, where $\s_{-i}$ denotes the strategy profile $\s$ without agent $i$'s strategy~$S_i$.

Every strategy profile $\s$ determines a \emph{created network} $G(\s)$ that consists of all the created edges, i.e., $G(\s) = (V,E(\s))$, where $E(\s) = \{\{u,v\} \mid u\in S_v \vee v\in S_u\}$. Note that we use the terms \emph{graph} and \emph{network} interchangeably. Given a created network $G(\s) = (V,E(\s))$, the \emph{cost} of agent $u\in V$ is defined as 
$$\cost_u(\s) = \alpha\cdot|S_u| + n-\left|N_2^{G(\s)}(u)\right|\,,$$ 
where $\alpha\in \mathbb{R}_{>0}$ is a fixed constant that is a parameter of our game. This parameter essentially determines the agents' trade-off between the cost of creating edges and the size of the 2-neighborhood. We mostly focus on $\alpha > 1$, since this models that no agent alone is important enough to justify the creation of an edge, which seems like a realistic assumption in large social networks.

We often express strategy profiles in form of a \emph{directed network} where we have a directed edge $(u,v)$ in the network if and only if $u$ created an edge to $v$, i.e., $v \in S_u$. The \emph{social cost} of any given strategy profile $\s$ is the sum of the agents' costs in $\s$, i.e., $\cost(\s) = \sum_{u\in V}\cost_u(\s)$. Moreover, for a given number of agents $n$ and parameter $\alpha$, let $\s^*_{n,\alpha}$ denote a \emph{socially optimal} strategy profile, i.e., we have that $\s^*_{n,\alpha}$ is a strategy profile that minimizes the social cost for given $n$ and $\alpha$.

We say that strategy $S_u'$ is an \emph{improving response} for agent $u$ in strategy profile $\s=(S_u,\s_{-u})$, if $\cost_u((S_u',\s_{-u})) < \cost_u((S_u,\s_{-u}))$. Strategy $S_u'$ is called a \emph{best response} to $\s_{-u}$, if $\cost_u((S_u',\s_{-u})) \leq \cost_u((S_u'',\s_{-u}))$, for any strategy $S_u'' \subseteq V\setminus\{u\}$, i.e., if no improving response of agent $u$ exists in strategy profile $(S_u',\s_{-u})$. Let $\mathcal{C} = \s_0,\s_1,\dots,\s_k,\s_0$ be a cyclic sequence of strategy profiles. Then $\mathcal{C}$ is called an \emph{improving response cycle (IRC)} if for any two neighboring strategy profiles $\s_i,\s_j$ we have that $\s_j$ is a strategy profile resulting from some agent in $\s_i$ changing her strategy to an improving response. It is well-known that the existence of an IRC is equivalent to the game not having the \emph{finite improvement property}, which is equivalent to the non-existence of an ordinal potential function~\cite{monderer1996potential}.

We say that $\s = (S_1,\dots,S_n)$ is a \emph{pure Nash equilibrium (NE)} if $S_i$ is a best response for all $1\leq i\leq n$, i.e., if no agent has an improving response in strategy profile $\s$. Given a NE $\s$, we call the corresponding created network $G(\s)$ \emph{stable} or \emph{Nash stable} and we say that $G(\s)$ is \emph{in equilibrium}. 
We also consider a weaker solution concept, where only certain strategy changes are allowed. For this, we say that for strategy profile $\s = (S_u,\s_{-u})$ agent $u\in V$ has a \emph{greedy improving response} $S_u'$, if $\cost_u((S_u',\s_{-u})) < \cost_u((S_u,\s_{-u}))$ and if either $S_u' = S_u \cup \{x\}$, or if $S_u' = S_u\setminus\{y\}$, or if $S_u' = (S_u\setminus\{y\}) \cup \{x\}$, for some $x\notin S_u$ and $y\in S_u$. A greedy improving response is an improving strategy change that consists of either adding, deleting, or swapping a single element in the agent's current strategy. We say that $\s$ is a \emph{greedy equilibrium (GE)}~\cite{lenznerGreedy2012} if no agent has a greedy improving response with respect to $\s$. By definition, we have that any NE is a GE, and we will show that the converse does not hold.

We want to investigate the impact of selfish behavior on the social cost. For this, let NE$(n,\alpha)$ denote the set of all strategy profiles that are in NE for a given number of nodes $n$ and parameter $\alpha$. We define the \emph{Price of Anarchy (PoA)} as the worst-case ratio of the social cost of any NE and the corresponding social optimum strategy profile. We follow standard literature in network creation games and analyze the PoA in dependence on the number of agents $n$ and edge cost $\alpha$~\cite{CP05}. Formally, $\text{PoA}(n,\alpha) = \max_{\s\in \text{NE}(n,\alpha)}  \frac{\cost(\s)}{\cost(\s^*_{n,\alpha})}$. In a similar vein, the \emph{Price of Stability (PoS)} considers the social cost of the best NE, i.e., $\text{PoS}(n,\alpha) = \min_{\s\in \text{NE}(n,\alpha)}  \frac{\cost(\s)}{\cost(\s^*_{n,\alpha})}$. The \emph{greedy Price of Anarchy} and \emph{greedy Price of Stability} are defined analogously for GE, i.e., $\text{gPoA}(n, \alpha) = \max_{\s\in \text{GE}(n,\alpha)}  \frac{\cost(\s)}{\cost(\s^*_{n,\alpha})}$ and $\text{gPoS}(n, \alpha) = \min_{\s\in \text{GE}(n,\alpha)}  \frac{\cost(\s)}{\cost(\s^*_{n,\alpha})}$.

\subsection{Related Work}
Game-theoretic models for the formation of networks have been a core topic in algorithmic game theory for almost three decades. Starting with the seminal work by Jackson and Wolinsky~\cite{JW96}, many variants of such models have been proposed and rigorously analyzed. In~\cite{JW96} $n$ selfish agents establish links among each other to maximize their individual benefit. The mutual benefit of a node pair $u,v$ in the formed network decreases with their distance. Our model can be seen as a simplified binary variant, where only nodes in distance~1 and 2 give a benefit of $1$ and all the others provide benefit~$0$. The PoA for certain parameter ranges of the connections model was analyzed by Baumann and Stiller~\cite{BS08}.

Later, the network creation game (NCG)~\cite{fabrikantNetwork2003} was proposed and it became the basis of almost all newer models, including our model. In the NCG, the agents buy incident edges for a price of $\alpha$ per edge and their goal is to minimize their closeness centrality, i.e., their sum of shortest-path distances to all other nodes. It is shown that NE always exist, being either a spanning star or a clique, that computing a best response strategy is NP-hard, and that the price of anarchy is bounded by $O(\sqrt{\alpha})$. Later, in a series of works~\cite{albersNash2014,demainePrice2007,MS10,MMM13,GHLL16,AM19,BL20,AB23} this bound has been shown to be constant for almost any value of $\alpha$, with the best general upper bound being $2^{O(\sqrt{\log n})}$~\cite{demainePrice2007}.

Many variants of the NCG have been studied, e.g., versions that involve cooperation~\cite{CP05,AFM09,FGLZ23}, non-uniform edge cost~\cite{MeiromMO14,CLMH14,CLMM17,BiloFLLM21}, robustness~\cite{MeiromMO15,CLMM16,Goyal16,Echzell0LM20}, geometric aspects~\cite{MoscibrodaSW06a,EidenbenzKZ06,bilo2019}, social networks~\cite{BiloFLLM21}, social distancing~\cite{social_distancing}, temporal networks~\cite{temporal_NCG,temporal_non_local}, or variants featuring locality~\cite{CL15,BiloGLP16,bilo_traceroute}. 
Close to our model is the bounded distance NCG by Bilò, Gualà, and Proietti~\cite{BiloGP15}. 
There, the agents have to ensure that all other nodes are in a given maximum or average distance $D$. For $D=2$ this is similar to our model with the main difference, that in our model the agents face a trade-off between their total cost for buying edges and the number of agents that are not within distance $D$. This can be seen as a soft constraint. For $D=2$ the authors show a PoA in $\Omega(\sqrt{n})$ and $O(\sqrt{n\log n})$.

The closest related NCG variant is the celebrity game by {\`A}lvarez, Blesa, Duch, Messegu{\'e} and Serna~\cite{ABDMS16celebrity}.
In this framework, each agent has an individual weight, and there is a general distance threshold denoted as $\beta$. The cost function for an agent is the sum of the edge cost and the total weight of agents located at a distance greater than $\beta$. Consequently, our model represents a specific instance of their star celebrity games with $\beta=2$ where all agents have a uniform weight of $w_{min}=w_{max}=1$. However, unlike the star celebrity games in \cite{ABDMS16celebrity}, we mostly focus on $\alpha > 1$, which implies the absence of "celebrities" in our model, i.e.\ nodes with weight larger than $\alpha$.
The authors of~\cite{ABDMS16celebrity} show the existence of stable networks for any $\alpha$. Further, they provide a bound on the diameter of NE of at most $2\beta+1$. For $\beta=2$ we improve this to $3$ and we show that this is tight. Moreover, it is shown that the PoS is 1 and that the PoA is upper bounded by $\bigO(\min\{\frac{n}{2},\frac{n}{\alpha}\})$. Both bounds carry over to our model. 
For $\beta=2$, they show a lower bound on the Price of Anarchy (PoA) that is linear in $n$. However, this result crucially relies on inconsistent tie-breaking when agents are indifferent about purchasing an edge. 
Besides bounds on the PoA, in~\cite{ABDMS16celebrity} it is also shown that computing a best response is NP-hard. This proof uses agents having non-uniform weights, which means that it does not apply to our setting, either.
Lastly, the focus on the 2-neighborhood is also central to the model by Anshelevich, Bhardwaj, and Usher~\cite{AnshelevichBU15}. In their model, agents strategically allocate effort shares to edges in a given network, with utility depending on the invested effort in their 2-neighborhood. This contrasts strongly with our model.

\subsection{Our Contribution}
We propose a 2-neighborhood maximization game, where selfish agents aim to have as many other agents as possible in their 2-neighborhood while keeping incurred costs low. Motivated by the setting of~\cite{ABDMS16celebrity}, we analyze a well-motivated and natural special case of the star celebrity game which allows us to extend their work by providing intriguing and more detailed results.

In previous works, it has been evident that the diameter of an equilibrium, i.e., the longest distance between any two agents in an equilibrium, increases in the parameter $\alpha$. 
We show that in our model the diameter has a constant and tight upper bound of 3 for NE and 4 for GE even for large $\alpha$, thereby improving the previous upper bound of 5 for NE established in~\cite{ABDMS16celebrity}.
We complete our findings on the properties of equilibria by giving insight into structural results, including a complete characterization of the existence of paths and cycles as substructures. We show, in particular, that Nash stable networks cannot contain a path of 3 or more consecutive nodes with a degree of 2, except for the $4$- and $5$-cycle (optionally with an additional leaf node) for certain small values of $\alpha$.

Moreover, we analyze the inefficiency and derive a bound on the PoA of $\Omega\left(\log\left(\frac{n}{\alpha}\right)\right)$ for NE. To do so, we construct a stable scalable network with nodes having arbitrarily large outdegrees. 
For GE with edge cost $1 \leq \alpha \leq 2$, we show an asymptotically tight $\Omega\left(n\right)$ lower bound on the gPoA by constructing a network consisting of three sets of agents that connect in a schematic way. Further, we show NP-completeness of calculating a best response and the existence of improving response cycles.
Finally, regarding future work, we generalize our diameter bounds to the general star celebrity setting with equal weights.

\section{Existence and Structure of Equilibria}\label{sec:Existence&Structure}
In this section, we give an insight into the richness of equilibria in 2-NMG, providing necessary properties that equilibria must satisfy, including a constant upper bound on the diameter. Formally, the \emph{diameter} of $G$, denoted as $\mathrm{diam}(G)$, is the maximum hop-distance on a shortest path between any two nodes of $G$. Among these equilibria, we highlight unique types with specific structures, such as paths consisting of degree two nodes. Moreover, as shown in \cite{ABDMS16celebrity}, Nash stable networks exist for any value of $\alpha$.
While equilibria with a diameter of $2$ can be found fairly easily for any $\alpha$, the analysis is much more complex for equilibria with larger diameters.

\subsection{Variety and Structure of Stable Networks}
First, we present a series of simple stable networks that will reveal a special structure in the end of this section. We prove their stability in the appendix.

\begin{restatable}{observation}{Zweierketten}
\label{obs:zweierketten}
    Let $G(\mathbf{s})=(V, E(\mathbf{s}))$ be a graph with \begin{enumerate}[label={(\alph*)}]
        \item \label{obs:eine_zweierkette} \mbox{$V=\{v_1,\dots, v_6\}$} and $\mathbf{s}=\left(\{v_2\},\{v_3\},\{v_4\}, \emptyset,\{v_1,v_4\},\{v_1,v_4\}\right)$ as depicted in Fig.~\ref{fig:smallNE}(a). Then, $G(\s)$ is stable for~$\alpha~\leq~2$.
        \item \label{obs:zwei_zweierketten} $V=\{v_1,\dots,v_8\}$ and $\s=(\{v_2,v_8\},\{v_3\},\emptyset, \{v_3,v_5\}, \{v_1,v_6\},\{v_7\},\emptyset,\{v_4,v_7\})$ as depicted in Fig.~\ref{fig:smallNE}(f). Then, $G(\s)=(V,E(\s))$ is stable for $\alpha=2$. 
    \end{enumerate}
 
\end{restatable}

\begin{figure}
    \centering
    \includegraphics[width=0.9\linewidth]{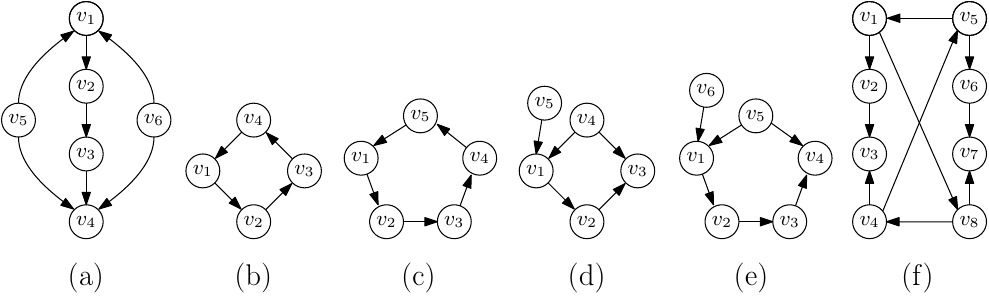}
    \caption{Illustration of NE for various $\alpha \leq 2$ with consecutive nodes of a degree of~$2$. The detailed descriptions can be found in Observations~\ref{obs:zweierketten} and~\ref{obs:NEs}.}
    \label{fig:smallNE}
\end{figure}
\noindent It is worth noting that in Fig.~\ref{fig:smallNE}(a), the identical nodes $v_5$ and $v_6$ can be replicated without affecting the stability of the graph. Additionally, we observe that both a 4-cycle and 5-cycle with and without a leaf (a node with degree $1$) can form stable networks for small values of $\alpha$ (see Fig.~\ref{fig:smallNE}(b) to (e)).

\begin{restatable}{observation}{Cycles}\label{obs:NEs}
    Let $G(\mathbf{s})=(V,E(\mathbf{s}))$ be a \begin{enumerate}[label={(\alph*)}]
        \item \label{obs:4cycle} 4-cycle with $V=\{v_1,\dots,v_4\}$ and \mbox{$\mathbf{s}=\left(\{v_2\},\{v_3\},\{v_4\},\{v_1\}\right)$} as depicted in Fig.~\ref{fig:smallNE}(b). Then, $G(\mathbf{s})$ is stable for $\alpha \leq 1$.
        \item \label{obs:5cycle} 5-cycle with $V=\{v_1,\dots,v_5\}$ and $\mathbf{s}=\left(\{v_2\},\{v_3\},\{v_4\},\{v_5\},\{v_1\}\right)$ as depicted in Fig.~\ref{fig:smallNE}(c). Then, $G(\mathbf{s})$ is stable for $\alpha \leq 2$.
        \item \label{obs:4cycleleaf} 4-cycle with a leaf with $V = \{v_1,\dots,v_4,v_5\}$ and $\mathbf{s}=\left(\{v_2\},\{v_3\},\emptyset,\{v_1,v_3\},\{v_1\}\right)$ as depicted in Fig.~\ref{fig:smallNE}(d). Then, $G(\mathbf{s})$ is stable for $\alpha=1$.
        \item \label{obs:5cycleleaf} 5-cycle with a leaf with $V=\{v_1,\dots,v_5,v_6\}$ and $\mathbf{s}=\left(\{v_2\},\{v_3\},\{v_4\},\emptyset,\{v_1,v_4\},\{v_1\}\right)$ as depicted in Fig.~\ref{fig:smallNE}(e). Then, $G(\mathbf{s})$ is stable for $\alpha=2$.
    \end{enumerate}
    
\end{restatable}

\noindent These equilibria have very special shapes and very few edges. Additionally, the outdegree, i.e., the number of bought edges is a small constant, even for $n \rightarrow \infty$ in Fig.~\ref{fig:smallNE}(a). Following, we show that the class of equilibria is much richer and we provide equilibria with superlinearly many edges and non-constant outdegree.

\begin{restatable}{theorem}{regularNE}
\label{thm:regular-NE}
    For all $\alpha$, there exists a stable network that has $n$ agents, a maximum outdegree in $\Omega(\log(\frac{n}{\alpha}))$, and $\Omega\left(n\log\left(\frac{n}{\alpha}\right)\right)$ edges.
\end{restatable}

\begin{proof}We start the proof by describing the construction of the network. An example can be seen in Fig.~\ref{fig:NE-k-regular-satellites}. For an arbitrary $k \in \mathbb{N}$, $k$ even, let $G'(\mathbf{s},k)=(V_B(k),E'(\mathbf{s},k))$ be a $k$-regular graph with $|V_B(k)|=2k +1$ base nodes. The orientation of the edges is arbitrary. Further, let
    \begin{equation*}
        \mathcal{K} := \left\{\{v_{z_1},v_{z_2},\dots,v_{z_{k+1}}\}\subset V_B(k): 1\leq z_1 < z_2 < \dots < z_{k} < z_{k+1} \leq |V_B(k)|\right\}
    \end{equation*} be the set of $(k+1)$-combinations of the nodes in $V_B$ and
        for every combination $\kappa \in \mathcal{K}$ we add new nodes $v_1^{\kappa},v_2^{\kappa},\dots, v_{\left\lceil \alpha \right\rceil}^{\kappa}$ to a set $V_S(k)$. This results in $\binom{2k+1}{k+1}$ different types of nodes added to $V_S(k)$ and $\left\lceil \alpha \right\rceil$ many of each type. We will refer to these as \textit{satellite} nodes resp. agents. Finally, with \begin{equation*}
            E(\mathbf{s},k)=E'(\mathbf{s},k) \cup \bigcup_{\kappa \in \mathcal{K}} \left\{ (v_i^{\kappa}, v): \text{ for all \,} i\in \left[1,\left\lceil \alpha \right\rceil \right] \text{ and } v \in \kappa\right\}
        \end{equation*} 
        and $V(k)=V_B(k)\cup V_S(k)$, the network $G(\mathbf{s},k)=(V(k), E(\mathbf{s},k))$ is fully constructed and consists of
         \begin{equation}
             |V(k)| = \left\lceil \alpha \right\rceil \cdot \binom{2k+1}{k+1} + 2k+1\label{eq:nodesRegularGraph}
         \end{equation}
         many nodes and
         \begin{equation}
             |E(\s,k)| = \left\lceil \alpha \right\rceil \cdot \binom{2k+1}{k+1} \cdot (k+1) + \frac{(2k+1)\cdot k}{2}\label{eq:edgesRegularGraph}
         \end{equation}
         many edges. The outdegree of the satellite nodes is $k+1$. 
    It remains to show that the $k$-regular graph with $2k+1$ nodes exists for even $k$, that the described network is indeed stable, and that $k \in \Omega\left(\log\left(\frac{n}{\alpha}\right)\right)$. 
    \paragraph{Existence of the base graph.} To prove the existence of a $k$-regular graph with $2k+1$ nodes, we begin with two cliques of size $k$. Next, we add a bipartite perfect matching between the nodes of these cliques, resulting in a connected graph with $2k$ nodes, where each node has degree $k$. By removing $\frac{k}{2}$ edges from the matching and connecting each node that loses an edge to a new node, we obtain a graph with $2k+1$ nodes, each of degree $k$.
 
\paragraph{Stability.} 
We prove stability by showing that (1.) the  network diameter is 2, and neither the base nodes (2.) nor the satellite nodes (3.) can improve.

\begin{enumerate}
    \item[1.] Since each $\kappa, \kappa' \in \mathcal{K}$ contains $k+1$ base nodes and there are $2k+1$ base nodes in the base graph, we have $\kappa \cap \kappa' \neq \emptyset$, ensuring that any two satellite nodes share at least one direct neighbor. Consequently, every base node has every satellite node in her 2-neighborhood. Moreover, any $k+1$-subset forms a dominating set in a $k$-regular graph with $2k+1$ many nodes. If this were not the case, there would exist a node $u$ such that $N(u) \cap \kappa = \emptyset$, implying $\mathrm{deg}(u) \leq 2k - |\kappa| = k-1$, and contradicting the construction. Thus, the 2-neighborhood of every node is complete and $\mathrm{diam}(G(\mathbf{s},k))=2$.

    \item[2.] Let $v \in V_B(k)$ be a base node with $\mathrm{deg}(v)=k$. Because of $\mathrm{diam}(G(\mathbf{s},k))=2$, there is no incentive for additional connections or swaps. The only candidates for beneficial changes are strategies with fewer bought edges. However, since each $k+1$-subset of $V_B(k)$ has $\left\lceil \alpha \right\rceil$ satellite nodes, exclusively connected to these base nodes, no connection between base nodes can be omitted.
    Therefore, no base node has an incentive to cancel edges, and no better strategy exists with fewer than $k$ edges, as each base node must be connected to $k$ other base nodes to ensure all satellite groups are in her neighborhood.
    
    \item[3.] By similar reasoning as in the last two paragraphs, satellite agents cannot improve their strategies. If an agent connects to fewer than $k+1$ base nodes, there will be at least one group of satellite nodes missing from her 2-neighborhood, since in a set of $2k+1$ nodes, two disjoint $k-$ and $k+1$-subsets always exist. No agent will accept missing any group of satellite nodes, as there are $\left\lceil \alpha \right\rceil$ of each type. Similarly, connecting to other satellite nodes instead of base nodes is not beneficial, as it leaves some satellite groups uncovered.
    \end{enumerate}

\paragraph{A lower bound on the outdegree.} 
The satellites have outdegree $k+1$ by construction. It remains to express this bound in terms of $n$. We have
\begin{align*}
    n=|V(k)|&=\left\lceil \alpha \right\rceil \cdot \binom{2k+1}{k+1} + 2k+1
    \leq \left\lceil \alpha  \right\rceil \cdot \binom{2k+2}{k+1}\,.
\end{align*}
By Stirling's approximation $x!\in\Theta\left(\sqrt{2\pi x}\left(\frac{x}{e}\right)^{x}\right)$, we obtain
\begin{align*}
    (2x)!&\in\Theta\left(\sqrt{2\pi(2x)}\left(\frac{2x}{e}\right)^{2x}\right)=\Theta\left(\sqrt{2}\sqrt{2\pi x} \cdot 2^{2x}\left(\frac{x}{e}\right)^{2x}\right)\,,
\end{align*}
and, by applying this to the binomial coefficient, we get
\begin{align*}
    \binom{2x}{x}&=\frac{(2x)!}{(x!)^2}
    \in \Theta\left(\frac{\sqrt{2} \cdot 2^{2x}\cdot\sqrt{2\pi x}\left(\frac{x}{e}\right)^{2x}}{\left(\sqrt{2\pi x}\cdot\left(\frac{x}{e}\right)^{x}\right)^2}\right)\,.
    \end{align*}
After extracting and canceling common terms, we derive an upper bound
\begin{align*}
    \frac{\sqrt{2} \cdot 2^{2x}\cdot\sqrt{2\pi x}\left(\frac{x}{e}\right)^{2x}}{\left(\sqrt{2\pi x}\cdot\left(\frac{x}{e}\right)^{x}\right)^2} &= \cfrac{\sqrt{2} \cdot 2^{2x}}{\sqrt{2\pi x}}\cdot\frac{\sqrt{2\pi x}\cdot\left(\sqrt{2\pi x}\left(\frac{x}{e}\right)^{2x}\right)}{\left(\sqrt{2\pi x}\cdot\left(\frac{x}{e}\right)^{x}\right)^2} \leq 2^{2x}\,.
\end{align*}
Finally, we use this bound on the binomial coefficient to express $k$ in terms of $n$, 
\begin{align*}
    \frac{n}{\left\lceil \alpha \right\rceil}&\leq \binom{2k+2}{k+1} \in \Omega\left(2^{2k+2}\right)\,,
\end{align*}which leads to
\begin{align}
       k&\in\Omega\left(\log_4\left(\cfrac{n}{\left\lceil \alpha \right\rceil}\right)\right)  = \Omega\left(\log\left(\frac{n}{\alpha}\right)\right)\label{eq:regularGraphOutDegree}
\end{align}
as lower bound for $k$. As seen in Equations~(\ref{eq:nodesRegularGraph}) and (\ref{eq:edgesRegularGraph}), the number of edges is $\Theta(k\cdot V(k))$, which also means it is in $\Omega(n \cdot k) = \Omega(n \cdot \log\left(\frac{n}{\alpha}\right))$.
\begin{figure}[t]
    \centering
    \includegraphics[scale=0.39]{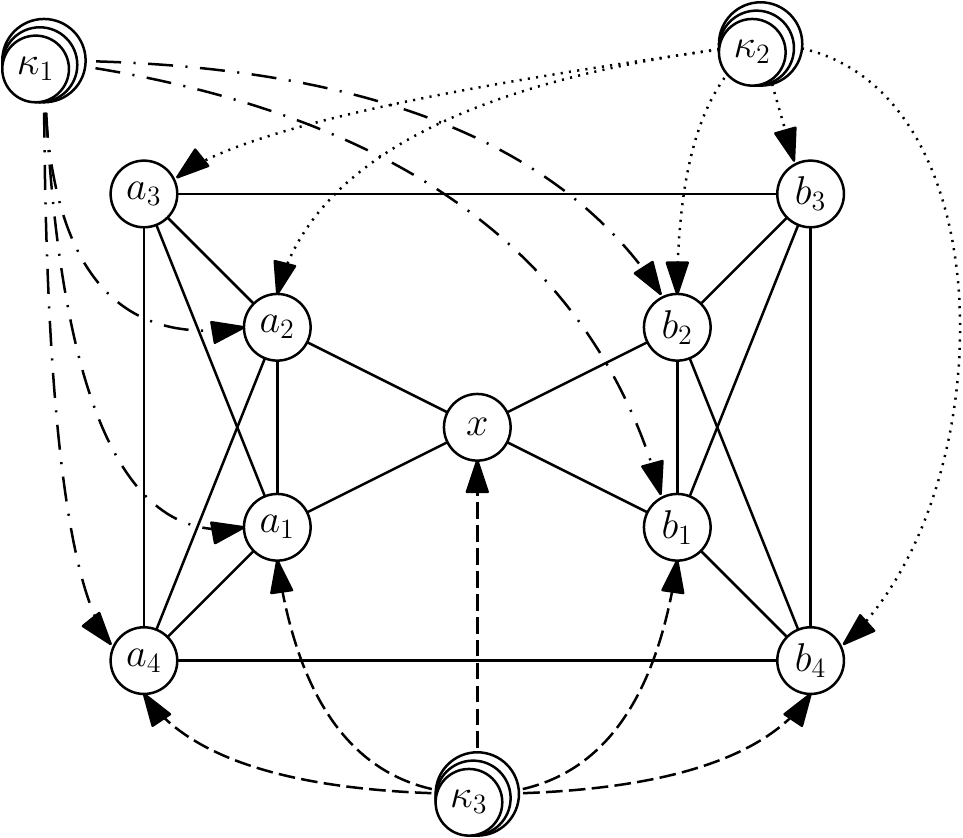}
    \caption{Illustration of a NE with $k=4$ (see Theorem~\ref{thm:regular-NE}). Out of all $\binom{9}{5}$ sets of satellites, three are drawn exemplarily.}
    \label{fig:NE-k-regular-satellites}
\end{figure}
\end{proof}

\noindent Clearly, a Nash stable network is also greedy stable. If we focus on greedy stable equilibria, we can even create networks with edges in the order of $\Theta(n^2)$.

\begin{restatable}{theorem}{GEdiamthreequadraticedges}
\label{thm:GE-diam-3-quadratic-edges}
    For $1\leq\alpha\leq 2$, there is a series of greedy stable networks with \mbox{$|E| \in \Theta(n^2)$}. 
    
\end{restatable}

\begin{figure}[t]
\centering\includegraphics[scale=0.4]{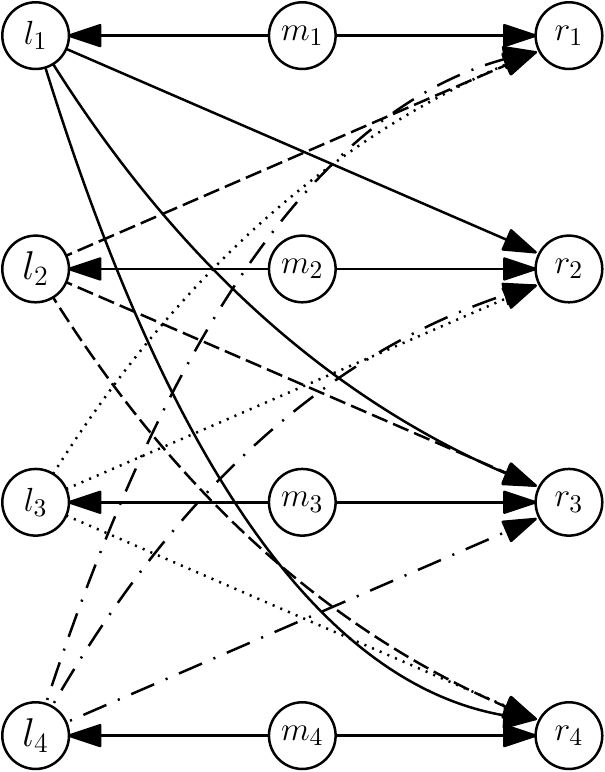}
\caption{A network used to show that GE with $|E| \in \Theta(n^2)$ exist (see Theorem~\ref{thm:GE-diam-3-quadratic-edges}).
}
\label{fig-GE-Diam3-Example}
\end{figure}
\begin{proof}
For a given $k\in\Na$, $k\geq3$, we will construct a greedy stable network $G(\s,k)$ with $|V(k)| = 3k$ nodes and $|E(k)| \in \Theta(k^2)$ edges. The graph  $G(\mathbf{s},k)$ is composed of three disjoint sets of nodes\textemdash$L$, $M$, and $R$, representing the \textit{left}, \textit{middle}, and \textit{right} sets, respectively\textemdash with each set containing $k$ nodes.
For $1\leq i \leq k$, the nodes $l_i\in L, m_i\in M, r_i \in R$ form triples. Each node $m_i$ builds two edges to $l_i$ and $r_i$, while $l_i$ builds $k-1$ edges to all $r_j\in R\setminus \{r_i\}$. Thus, for a given $k$, the graph has $3k$ nodes and $k\cdot(k+1)$ edges. An example for $k=4$ is depicted in Fig.~\ref{fig-GE-Diam3-Example}.
This construction has a diameter of 3 and we show in the appendix that it is indeed greedy stable. 
\end{proof}
\noindent To this point, we have provided a diverse range of equilibria, ranging from those with very few edges to Nash stable networks featuring a superlinear number of edges, as well as greedy stable networks characterized by a quadratic number of edges in $n$. Despite their distinct characteristics, these equilibria share underlying structural properties that we will explore in the remainder of this section.
Naturally, the question arises whether stable networks that are not empty may consist of several different connected components. {\`A}lvarez, Blesa, Duch, Messegu{\'e} and Serna~\cite{ABDMS16celebrity} answer this question negatively and their proof transfers also to greedy equilibria. Since the graph without edges is the only non-connected stable graph, it is sufficient to limit the analysis to graphs that consist of only one component.

As explained in Section~\ref{sec:Intro}, intermediaries hold significant importance in social networks. However, in sequences of nodes with a degree of 2, these intermediary nodes are limited in their capacity to fulfill this role since they only connect two of their neighbors. Consequently, it raises an important and intriguing question of whether such substructures can actually emerge in equilibria.

\begin{restatable}{theorem}{noLongPaths}
\label{thm:noLongPaths}
 Let $G(\mathbf{s})=(V,E(\mathbf{s}))$ be a Nash stable network for any $\alpha$. Then there is either no path of three or more consecutive nodes of degree 2 or \begin{itemize}
    \item $G(\mathbf{s})$ is a 4-cycle for $\alpha\leq 1$ (see Observation~\ref*{obs:NEs}\ref{obs:4cycle}),
    \item $G(\mathbf{s})$ is a 4-cycle with leaf for $\alpha=1$ (see Observation~\ref*{obs:NEs}\ref{obs:4cycleleaf}),
    \item $G(\mathbf{s})$ is a 5-cycle for $\alpha\leq 2$ (see Observation~\ref*{obs:NEs}\ref{obs:5cycle}),
    \item $G(\mathbf{s})$ is a 5-cycle with leaf for $\alpha=2$ (see Observation~\ref*{obs:NEs}\ref{obs:5cycleleaf}).
 \end{itemize}

\end{restatable}
 
\noindent We prove the theorem in the appendix by showing that there cannot be three or more consecutive nodes of degree 2, in particular no cycles, except for the few cases described in the theorem. Note that the stability of these were already shown in Observations~\ref*{obs:NEs}\ref{obs:5cycle} and \ref*{obs:NEs}\ref{obs:5cycleleaf}.
\noindent Next, we show that although equilibria with two consecutive nodes with a degree of 2 can indeed exist (see Observation~\ref*{obs:zweierketten}\ref{obs:eine_zweierkette}), for $\alpha \neq 2$ a stable network cannot contain two such pairs simultaneously.
\begin{restatable}{observation}{keinezweidegzweiketten}
\label{obs:keine2deg2Ketten}
    For a stable network $G(\mathbf{s})=(V,E(\mathbf{s}))$ and all $\alpha$, it holds 
    \begin{align*}
        \exists\,u, v_1,v_2,w\in V:  &\{u,v_1\},\{v_1,v_2\},\{v_2,w\} \in E(\s) \\
        &\wedge \mathrm{deg}(v_1),\mathrm{deg}(v_2)=2 \wedge \mathrm{deg}(u),\mathrm{deg}(w)\neq 2\\
        \Longrightarrow \nexists\,u', v_1',v_2', w'\in V\setminus\{v_1,v_2\}:& 
        \{u',v_1'\},\{v_1',v_2'\},\{v_2',w'\} \in E(\s) \\
           &\wedge \mathrm{deg}(v_1'),\mathrm{deg}(v_2')=2 \wedge \mathrm{deg}(u'),\mathrm{deg}(w')\neq 2\,,
    \end{align*}
    or $\alpha=2$ and $G(\s)$ has the structure described in Observation~\ref*{obs:zweierketten}\ref{obs:zwei_zweierketten}.
\end{restatable}
\noindent We defer the proof to the appendix. Note that the networks presented in Observations~\ref{obs:zweierketten} and \ref{obs:NEs} provide a complete picture on NE with consecutive nodes of degree 2. 
Theorem~\ref{thm:noLongPaths} and Observation~\ref{obs:keine2deg2Ketten} show that the only candidates for these NE are those shown in Fig.~\ref{fig:smallNE}. 
When analyzing the structural properties of equilibria, it makes sense to explore the occurrence of leaf nodes as their existence appears inconsistent with the analogy of social networks, too. These leaf nodes rely entirely on a single neighbor for connectivity, as well as for accessing and sharing information.
For the particularly intriguing case of $\alpha>1$, which closely resembles real-world scenarios, it can be observed easily that if leaf nodes are present, they may only be connected to the same node. The formal proof is deferred to the appendix.
\begin{restatable}{lemma}{blaetterlemma}\label{lem-no-two-leaves-at-different-nodes}
    In a stable network $G(\mathbf{s})=(V,E(\mathbf{s}))$ for $\alpha>1$, there is only one node to which all leaf nodes $L:=\set{v\in V}{\mathrm{deg}(v)=1}$ connect.
\end{restatable}
\begin{corollary}[Follows from Lemma~\ref{lem-no-two-leaves-at-different-nodes}]\label{cor:stable-trees-are-always-stars}
    For $\alpha>1$, every Nash equilibrium graph that is a tree is also a star. 
\end{corollary}

\subsection{Tight Bounds on the Diameter}
Although we have seen in the previous section that large networks with a particularly large number of edges exist, we derive tight constant bounds on the diameter of stable and greedy stable networks. Interestingly, the constant bound remains true even for large values of $n$ and $\alpha$, which contrasts with the results of Fabrikant, Luthra, Maneva, Papadimitriou, and Shenker~\cite{fabrikantNetwork2003}.

\begin{theorem}\label{thm:NE-diam-3}
    For $\alpha \geq 1$, all connected Nash stable networks $G(\mathbf{s})$ fulfill \linebreak \mbox{$\mathrm{diam}(G(\mathbf{s}))\leq 3$}. This bound is tight.
    
\end{theorem}

\noindent We prove the theorem by providing a proof for the upper bound in Lemma~\ref{lem:NE-at-most-diam-3} and a lower bound example in Lemma~\ref{lem-NE-diam3-1alpha3}.

\begin{lemma}\label{lem:NE-at-most-diam-3}
    For $\alpha\geq1$, all connected Nash stable networks $G(\mathbf{s})$ fulfill \linebreak \mbox{$\mathrm{diam}(G(\mathbf{s}))\leq 3$}.
\end{lemma}
\begin{proof}
    Suppose there is a stable network $G(\mathbf{s})=(V,E(\mathbf{s}))$ with \mbox{$\mathrm{diam}(G(\mathbf{s}))\geq 4$}. 
    Thus, $\exists\, v_1,v_5\in V$ with $d_G(v_1,v_5)=4$.
     Let $v_1$---$v_2$---$v_3$---$v_4$---$v_5$ denote a shortest path from $v_1$ to $v_5$ in $G$. In particular, $d_G(v_1,v_3)=d_G(v_3,v_5)=2$, $d_G(v_1,v_2)=1$, $d_G(v_2,v_5)=3$, $d_G(v_1,v_4)=3$, and $d_G(v_4,v_5)=1$. For $\alpha=1$, $v_1$ has an improving response in connecting to $v_4$ or $v_5$, since $v_4,v_5\notin N_2^G(v_1)$. It remains to prove the statement for $\alpha>1$.
     If $v_1$ had an incoming edge from a node $v\in V$, which means that the benefit to the player $v$ from this edge is at least $\alpha$, then it would be beneficial for $v_5$ to also connect to $v_1$, since $N_1^G(v_1)\cap N_1^G(v_5)=\emptyset$. As a result, player $v_1$ cannot have any incoming edges, and by symmetry, the same holds true for $v_5$. 
     Without loss of generality let \mbox{$\cost_{v_1}(\s)\geq \cost_{v_5}(\s)$}. The agent $v_1$ has an improving response by copying the strategy of agent $v_5$. This would add $v_1$ and $v_5$ to each other's 2-neighborhoods and agent~$v_1$ would possess new costs of $\cost_{v_1}(\mathbf{s'})=\cost_{v_5}(\mathbf{s})-1 < \cost_{v_5}(\mathbf{s}) \leq \cost_{v_1}(\mathbf{s})$.  
\end{proof}

\begin{restatable}{lemma}{lemNEdiamthreeonealphathree}
\label{lem-NE-diam3-1alpha3}
    There is a Nash stable network with diameter 3.
\end{restatable}
\begin{proof}
\begin{figure}
\centering
\includegraphics[scale=0.4]{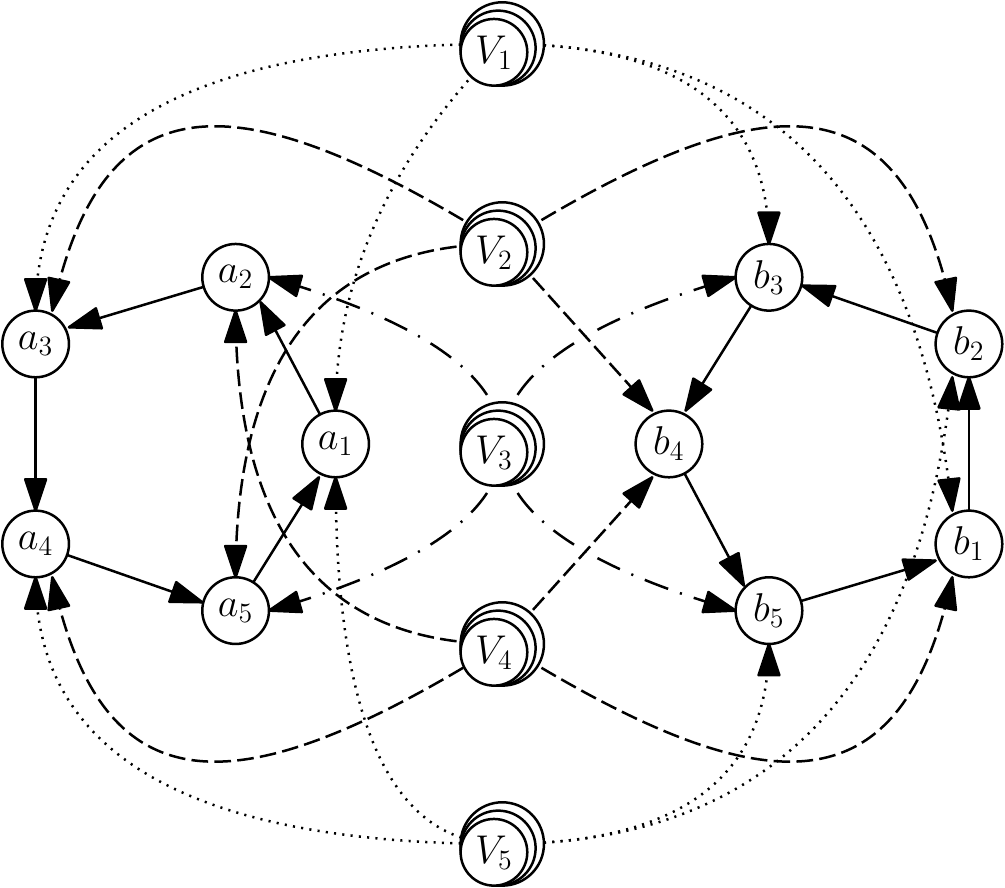}
\caption{Illustration of the NE used in Lemma~\ref{lem-NE-diam3-1alpha3}. Nodes with multiple borders represent multiple identical nodes.}
\label{fig-NE-Diam3-Example}
\end{figure}
We choose any $1 \leq \alpha<3$. 
First, we explain how to construct the network, then we show its Nash stability.
The stability of the network relies heavily on the properties of integer groups modulo 5. The network consists of the following set of nodes (see Fig.~\ref{fig-NE-Diam3-Example}): 
\begin{itemize}
    \item Two sets $A$ and $B$ of five nodes $a_1,\ldots,a_5$ resp. $b_1,\ldots,b_5$, forming two 5-cycles
    \item Five sets $V_1,\ldots,V_5$ of $\lceil\alpha\rceil$ nodes each
\end{itemize}
The network contains the following edges for $1 \leq i \leq 5$, $i\in \mathbb{N}$:
\begin{itemize}
    \item $a_i$ and $b_i$ have both one edge to $a_{i+1 \bmod 5}$ and $b_{i+1 \bmod 5}$ (the node on the 5-cycle with the next higher index).
    \item Every node in $V_i$ has edges to $a_{2i-1\bmod5}$ and $a_{2i+1\bmod5}$.
    \item Every node in $V_i$ has edges to $b_{i\bmod5}$ and $b_{i+2\bmod5}$.
\end{itemize}
The resulting graph has diameter 3, e.g., because $a_1$ and $b_4$ have no common neighbor. A detailed stability proof can be found in the appendix, where we also demonstrate that this construction is greedy stable for any $\alpha \geq 1$.  
\end{proof}
\begin{observation}\label{obs:alpha<2_diam<=2}
    For $\alpha<1$, all connected (greedy) stable networks $G(\s)$ fulfill $\mathrm{diam}(G(\s))\leq 2$ as it is beneficial to connect to any node that is not in a respective 2-neighborhood.
\end{observation}

\noindent As the diameter cannot exceed the bound of $3$, even for large values of $\alpha$, it is worth examining whether similar results apply to greedy equilibria. Indeed, we establish a tight bound of $4$ for greedy stable networks when $\alpha\geq 3$.

\begin{restatable}{theorem}{GEdiamFour}\label{thm:GE-diam-4}
    For $\alpha\geq 3$, all connected greedy stable networks $G(\mathbf{s})$ fulfill \linebreak\mbox{$\mathrm{diam}(G(\mathbf{s}))\leq 4$.} This bound is tight.
\end{restatable}
\noindent We again prove the theorem by providing an upper and a lower bound in the appendix.
Furthermore, the network used in the proof is only stable for $\alpha\geq3$; in fact, we show in the appendix that for $\alpha<3$, no greedy equilibrium with diameter $4$ exists, and prove a diameter of $3$ as a tight bound.

\begin{restatable}{theorem}{GEdiamThree}\label{thm:GE-diam-3}
    For $1\leq\alpha<3$, all connected greedy stable networks $G(\mathbf{s})$ fulfill \linebreak \mbox{$\mathrm{diam}(G(\mathbf{s}))\leq 3$}. This bound is tight.
\end{restatable}

\section{Dynamics \& Hardness}
After exploring the richness of (greedy) equilibria by presenting stable networks with diverse characteristics and highlighting structural properties, we now turn to analyzing best-response dynamics. However, we will show that computing best responses is NP-complete and that the finite improvement property does not hold. To start, we demonstrate that the 2-NMG fails to satisfy the finite improvement property as there exists an improving response cycle.
\begin{restatable}{proposition}{dynamics}
\label{prop:dynamics}
    For any $\alpha$, there exists a 2-NMG instance with cyclic sequence of strategy profiles of improving responses.
\end{restatable}
\noindent In Fig.~\ref{fig:IRC-Example} we illustrate an improving response cycle that holds for any edge price $\alpha$, as only edge-swaps are performed. 
Dashed edges represent the edges from $a_1$, $a_2$, and $b$, which change at each step. 
Full details of the construction and all six strategy changes appear in the appendix.

\begin{figure}[t]
	\centering
	\subfloat[$\s_0$]{
		\resizebox{.26\linewidth}{!}{
\includegraphics[]{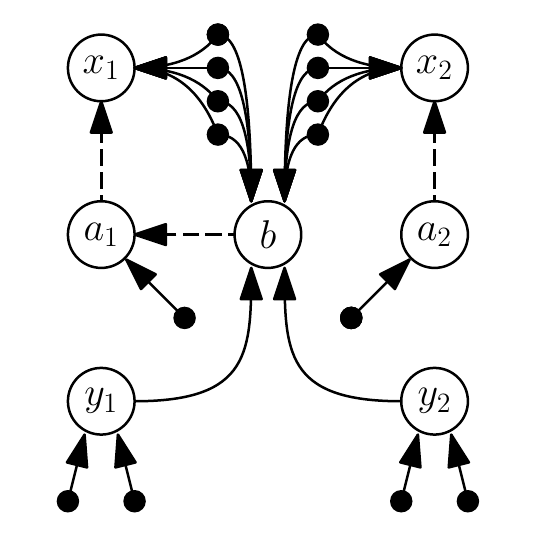}
		}
	}
	\hfil
	\subfloat[$\s_1$]{
		\resizebox{.26\linewidth}{!}{
			\includegraphics[]{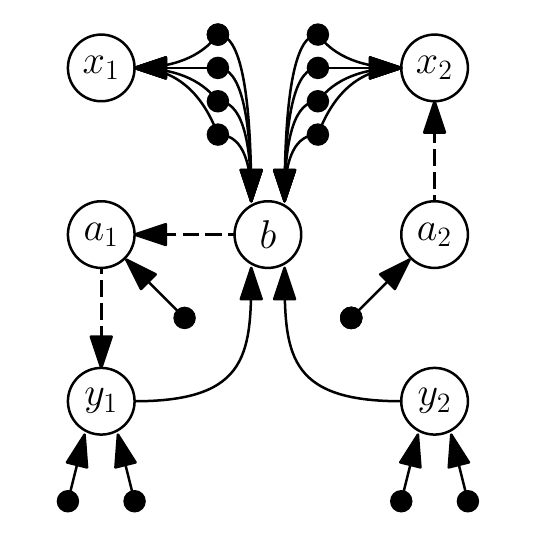}
		}
	}
	\hfil
	\subfloat[$\s_2$]{
		\resizebox{.26\linewidth}{!}{
            \includegraphics[]{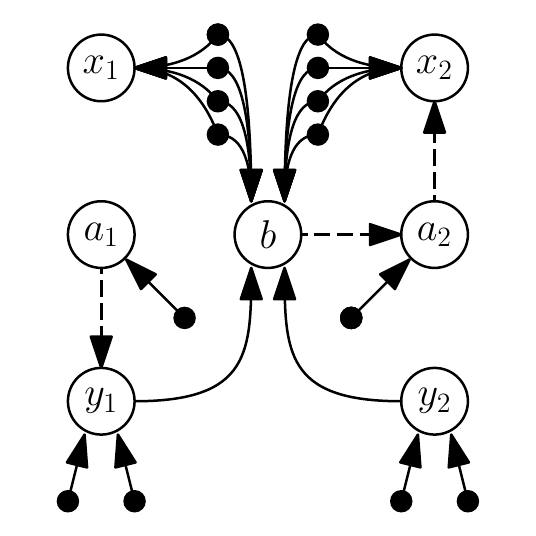}
		}
	}
	\hfil
	\subfloat[$\s_3$]{
		\resizebox{.26\linewidth}{!}{
			\includegraphics[]{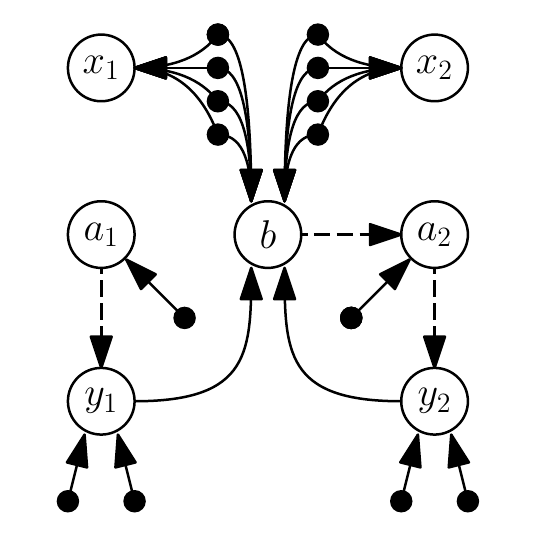}
		}
	}\hfil
	\subfloat[$\s_4$]{
		\resizebox{.26\linewidth}{!}{
			\includegraphics[]{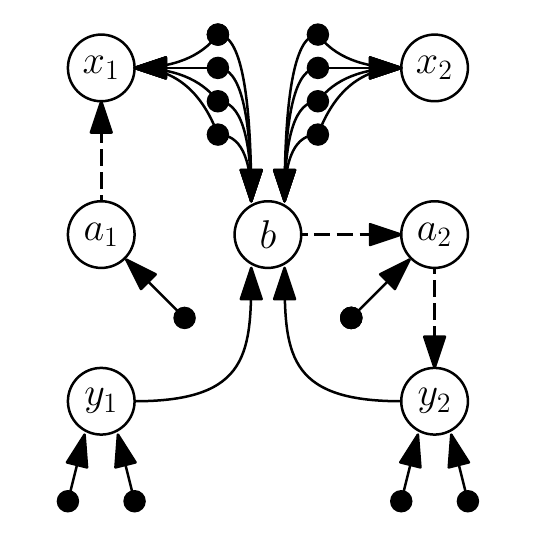}
		}
	}\hfil
	\subfloat[$\s_5$]{
		\resizebox{.26\linewidth}{!}{
			\includegraphics[]{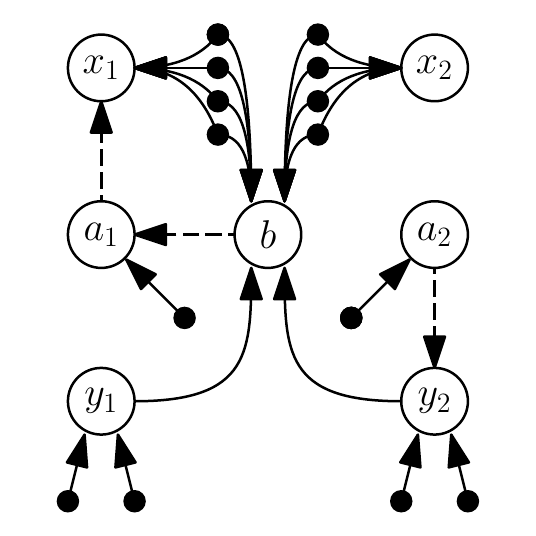}
		}
	}
	\caption{Improving response cycle for the 2-NMG for any edge price $\alpha$.}
    \label{fig:IRC-Example}
\end{figure}

Next, we prove that computing a best response in the 2-NMG is NP-complete, using a reduction from the well-known Dominating Set (DS) problem \cite{garey1979computers}. 
To show NP-hardness, we reframe the optimization problem of computing a best response as a corresponding decision problem, then reduce DS to it. Note that the reduction in \cite{ABDMS16celebrity} does not apply here, as it assumes weights of 2, while our setting implies weights of 1. This, along with the relaxation to $\alpha>0$, makes a new approach necessary. 
Therefore, we offer a reduction where we carefully modify the given graph by creating multiple copies of all nodes. Complete details of the proof and for NP-membership are provided in the appendix.

\begin{restatable}{theorem}{hardness}
\label{thm:best-Response-NPH}
Computing a best response strategy for a given agent $v$ is NP-complete for the 2-NMG with edge cost $\alpha>0$.
\end{restatable}

\section{Price of Anarchy}

We now establish lower bounds for the (greedy) price of anarchy (g)PoA. The upper bound of $\frac{n}{\alpha}$ is derived from the work on celebrity games by {\`A}lvarez, Blesa, Duch, Messegu{\'e} and Serna~\cite{ABDMS16celebrity}. Their lower bound of $\bigOmega\left(\frac{n}{\alpha}\right)$ 
does not transfer to our setting, due to the need of different node weights. To determine the lower bounds, we note that in the 2-NMG, the star graph is optimal for $\alpha\leq n$, the empty graph is optimal for $\alpha\geq n$, and both are optimal when $\alpha = n$. This follows immediately from Proposition~2 in~\cite{ABDMS16celebrity}.

To analyze the (greedy) price of anarchy, we focus on connected networks, as stated earlier, given that the only non-connected equilibrium is the empty graph for any $\alpha$. Its social cost equals $n(n-1)$ independent of $\alpha$, yielding tight (g)PoA bounds of $\max\{\frac{n}{\alpha},1\}$ when compared to the respective socially optimal states.
For Nash equilibria, we provide a lower bound of $\log\left(\frac{n}{\alpha}\right)$ for any $\alpha$, whereas for greedy equilibria, we establish a tight bound of $\Theta(n)$ for $1 \leq \alpha \leq 2$.
\begin{theorem}\label{thm:PoA-NE-lowerbound-diam-2}
    For the 2-NMG with arbitrary $\alpha$, there are connected Nash equilibria with diameter 2 and cost of at least $\log\left(\frac{n}{\alpha}\right)$ times the cost of an optimal network. This results in a PoA of $\Omega\left(\log\left(\frac{n}{\alpha}\right)\right)$.
\end{theorem}
\begin{proof}
    Theorem~\ref{thm:regular-NE} shows the existence of a Nash stable network for any $\alpha$, which we will use to derive a lower bound of $\bigOmega\left(\log\left(\frac{n}{\alpha}\right)\right)$ for the PoA. Since the diameter is 2 and therefore all agents' 2-neighborhoods are complete, the social cost depends solely on the number of edges purchased.
    For a given $\alpha$ and $k$, the social cost is
\begin{equation*}
    cost(\s)=\alpha\cdot |E(k)| = \alpha\cdot\left( \left\lceil \alpha \right\rceil \cdot \binom{2k+1}{k+1} \cdot (k+1) + (2k+1)\cdot k\right)\,,
\end{equation*}
and, as earlier shown, the number of nodes is
\begin{equation*}
n=|V(k)|=\left\lceil \alpha \right\rceil \cdot \binom{2k+1}{k+1} + 2k+1\,.
\end{equation*}
Comparing the described network $G(\s,k)$ with the optimal network (a star with $n$ nodes), we obtain a lower bound on the PoA as follows:

\begin{align*}
    \text{PoA}(k,\alpha) &\geq \cfrac{cost(\s)}{cost(\s_{OPT})} \geq \cfrac{\alpha|E|}{\alpha(n-1)} = \cfrac{ \left\lceil \alpha \right\rceil \cdot \binom{2k+1}{k+1} \cdot (k+1) + (2k+1)\cdot k}{\left\lceil \alpha \right\rceil \cdot \binom{2k+1}{k+1} + 2k} \\
    &\geq k \left(\frac{\lceil \alpha \rceil \cdot \binom{2k+1}{k+1} + 2k + 1}{\lceil \alpha \rceil \cdot \binom{2k+1}{k+1} + 2k}\right) \geq k\,.
\end{align*}
From Equation~(\ref{eq:regularGraphOutDegree}) of Theorem~\ref{thm:regular-NE}, we know $k\in\Omega\left(\log\left(\frac{n}{\alpha}\right)\right)$, thus the PoA is in the order of $\bigOmega\left(\log\left(\frac{n}{\alpha}\right)\right)$ as well. 
\end{proof}
\noindent Note that, unlike the approach in~\cite{ABDMS16celebrity}, our construction does not involve inconsistent tie-breaking.
As Nash equilibria are greedy equilibria, this lower bound also holds for the gPoA.
The upper bound of $\frac{n}{\alpha}$ shown in Lemma 2 of {\`A}lvarez, Blesa, Duch, Messegu{\'e} and Serna~\cite{ABDMS16celebrity} also holds for greedy equilibria, as it only considers the potential benefit provided by each single edge.
We now enhance the lower bound presented in Theorem~\ref{thm:PoA-NE-lowerbound-diam-2} for $1\leq \alpha\leq 2$, improving it to a linear gPoA lower bound of $\bigOmega(n)$, which asymptotically matches the upper bound of~$\frac{n}{\alpha}$.

\begin{theorem}\label{thm:PoA-GE-lowerbound-diam-3}
     For the 2-NMG with $1\leq \alpha \leq 2$, connected greedy equilibria exist with a diameter of 3 and their cost is in the order of $\Theta(n)$ times the cost of an optimal network. 
\end{theorem}

\begin{proof}
    Theorem~\ref{thm:GE-diam-3-quadratic-edges} shows that for $1\leq \alpha \leq 2$ and $k\in \Na$ with $k\geq3$, greedy stable networks exist with $3k$ nodes and $k(k-1)+2k = k^2+k$ edges. In terms of $n=|V|$, we obtain
    \[|V|=n=3k\;\text{ and }\;|E|=k^2 + k =\left(\frac{n}{3}\right)^2+\frac{n}{3}\,.\]
    The number of edges alone results in a gPoA of order $\Theta(n)$ when compared to the $n-1$ edges in the star social optimum. 
    In fact, since the networks of Theorem~\ref{thm:GE-diam-3-quadratic-edges} also contain $|M|=\frac{n}{3}$ nodes in the middle column, each of which is at distance 3 from all other nodes in $M$, the actual social cost doubles. However, this does not affect the asymptotic social cost.
\end{proof}

\section{Future Work \& Conclusion}
We considered the 2-neighborhood maximization game, a model that naturally aligns with social networks, where an agent benefits typically stem from her direct neighbors and their connections, rather than from more distant or unrelated individuals. In this context, we focused on the structure of both Nash equilibria (NE) and greedy equilibria (GE), with particular emphasis on the quality of these equilibria.
While we established constant upper bounds on the diameter, independent of $\alpha$ and $n$, and present Nash and greedy equilibria of different qualities, the key open question remains whether the PoA of NE is indeed asymptotically smaller than the PoA of GE.

Also, in line with the celebrity games in~\cite{ABDMS16celebrity}, we believe the extension to \mbox{general $\beta>2$} to be an interesting and promising direction for future research. By carefully reviewing Lemma~\ref{lem:NE-at-most-diam-3} and Observation~\ref{obs:alpha<2_diam<=2} for NE resp. Lemma~\ref{lem:no-GE-alpha>=5} and Theorem \ref{thm:GE-diam-3} for GE, it can be shown that our diameter bounds extend to \mbox{general $\beta\geq2$}, following the same proof techniques. This yields the following:

\begin{theorem}
    \label{thm:allgemeineBeta_diamBounds}
    For $\beta\geq2$, all connected Nash or greedy stable networks $G(\s)$ fulfill \begin{equation*}
     \mathrm{diam}(G(\s)) \leq 
    \begin{cases}
    \beta & \text{for } \alpha<1\,, \\
    \beta +1 & \text{for } \alpha=1\,, \\
    2\beta-1 & \text{for } \alpha>1 \text{ and $\s$ is NE\,,} \\
    2\beta & \text{for } \alpha>1 \text{ and $\s$ is GE\,.} 
    \end{cases}
    \end{equation*}
\end{theorem}

\noindent With this, we refine the bound of $2\beta+1$ given by {\`A}lvarez, Blesa, Duch, Messegu{\'e} and Serna~\cite{ABDMS16celebrity}, improving it to $2\beta -1$, which is tight for $\beta=2$.

Furthermore, we find it worthwhile to investigate an analogue to the well-known tree conjecture \cite{fabrikantNetwork2003}. Specifically, what is the largest value of $\alpha$, parameterized by the number of nodes~$n$, for which an equilibrium exists that is neither a star nor the empty graph?

\paragraph{Acknowledgements.}This work was supported by the Federal Ministry of Education and Research (BMBF), grant number 03SF0687A, and by the DFG project GEONET under grant DFG 442003138.

\printbibliography

\newpage
\section*{Appendix}

\Zweierketten*
\begin{proof}
    \ref{obs:eine_zweierkette}: Since $\mathrm{diam}(G(\mathbf{s}))=2$, no agent has any incentive to establish a new connection. So the only options that can be considered are either deleting edges and/or swapping edges. It is clear that $v_4$ would not deviate from her strategy \mbox{$S_{v_4}=\emptyset$}. None of $v_1,v_2$ or $v_3$ has an incentive to delete their edge, as they would lose two, and in the case of $v_3$, even 3 nodes from their 2-neighborhood, but would only save $\alpha\leq2$. Finally, $v_5$ and $v_6$ have the same strategy. If $v_5$ deletes one edge and connects only to one other agent, then, regardless of the agent she chooses, she loses $2$ nodes from her 2-neighborhood, but saves only $\alpha\leq2$. The same holds for $v_6$ and even any additional copy of them.
    
    \ref{obs:zwei_zweierketten}:
    First, we observe that $v_1,v_4,v_5,$ and $v_8$ have all agents in their 2-neighborhood, while $\{v_2,v_3\}$ have all agents except $\{v_6,v_7\}$ and vice versa.
    Due to symmetries in the instance, there are three groups of agents with identical behavior: agents $v_1,v_4,v_5,$ and $v_8$; agents $v_2$ and $v_6$; and agents $v_3$ and $v_7$.
    Regardless of the group, no agent has an incentive to establish an additional connection, as none is missing 3 or more agents.
    Also, no agent wants to remove any edges either, as it would reduce their 2-neighborhood by at least $2$ agents. It remains to show that swaps aren't beneficial either and we start with the agents $v_2$ and $v_6$. To improve their strategy they need to connect in a way that increases their 2-neighborhood by $3$ or more agents. Since deletion and addition of an edge are not beneficial, the only other option would be to swap their one edge. Since there is no node that has at least three nodes of $v_3,v_4,v_6,$ and $v_7$ as direct neighbors, no beneficial swap is available and $v_2$ is stable. Through symmetry, the argument works analogous for $v_6$.

    Since the agents $v_3$ and $v_7$ do not establish connections themselves and are already stable, the only possible swaps are those initiated by agents $v_1,v_4,v_5,$ and $v_8$. As their 2-neighborhoods are already complete, it remains to show that the combination of deleting one edge and swapping the other is no improving response either.
    We will analyze agent $v_1$ as a representative of this group of agents. Because of the edge $(v_5,v_1)$ agents $v_4,v_5$ and $v_6$ are already in her 2-neighborhood. As there is no way to establish an edge in a way to gain 3 or more of the remaining agents into her 2-neighborhood, there is no improving response.
    Thus, the presented instance is indeed Nash stable.
\end{proof}

\Cycles*

\noindent The following proofs for the 4-cycle and 5-cycle, resp. 4-cycle and 5-cycle with a leaf work in similar ways, so we combined them. The argumentation in parentheses refers to the respective 5-cycle.

\begin{proof}
\ref{obs:4cycle},\ref{obs:5cycle}:
    Every agent establishes one edge.
    Additional connections are not beneficial, as the diameter of $G(\mathbf{s})$ is 2. Removing the edge is also not favorable, as one would lose $1$ (resp. $2$) other agents from her neighborhood but only save $\alpha\leq 1$ (resp. $\alpha\leq 2$). 

    \ref{obs:4cycleleaf},\ref{obs:5cycleleaf}:
    As all nodes are missing at most 1 (resp. 2) other nodes in their 2-neighborhoods, adding additional edges is not a strict improvement. All agents except $v_4$ (resp. $v_5$) construct at most one edge, and because of the symmetry of the underlying cycle there are no agents where swapping to would increase the agent's 2-neighborhoods. 
    It remains to show that $v_4$ (resp. $v_5$), the only agent that establishes two edges, cannot improve her strategy. Since with both edges she exclusively gains at least one (resp. two) other agents into her 2-neighborhood, she has no incentive to deviate.
\end{proof}

\GEdiamthreequadraticedges*
\begin{proof}[Proof (continued)]
    To verify greedy stability, it suffices to ensure the following: each edge built by an agent connects them to at least $\lceil\alpha\rceil$ nodes, meaning deleting it is not beneficial; no agent can gain access to at least$\lceil\alpha\rceil$ nodes by building a new edge; and swapping an edge does not increase the agent's 2-neighborhood.
    As all triples $(l_i,m_i,r_i)$ in $G(\mathbf{s},k)$ are constructed with analogous edges, we only need to check stability for the nodes in one triple.
    \paragraph{Nodes in $L$:} Without loss of generality consider node $l_1$. She builds $k-1$ edges and has one incoming edge $(m_1,l_1)$. 
    Each of the $k-1$ edges $(l_1,r_j)$ with $2\leq j \leq k$ allows an exclusive connection to $r_j$ and $m_j$.
    Moreover, any of these edges $(l_1,r_j)$ allow two hop connections to all nodes in $L$ besides $l_j$.
    The incoming edge $(m_1,l_1)$ establishes a two hop connection from $l_1$ to $r_1$.
    Thus, $l_1$ has all other nodes in her 2-neighborhood. Adding or swapping a single edge is not beneficial. Neither is deleting a single edge $\{l_1,r_j\}$, as it would remove $2\geq\alpha$ nodes, i.e.\ $r_j$ and $m_j$, from the 2-neighborhood.
    
    \paragraph{Nodes in $M$:} The nodes in this set, such as $m_1$, build two edges and have no incoming edges. 
    The two edges are built to nodes in the same triple, i.e.\ $l_1$ and $r_1$. Over these edges, two hop connections to all nodes in $L$ and $R$ are ensured, but all other nodes in $M$ remain at distance 3.
    Due to the construction of $G(\mathbf{s},k)$, adding a single edge is not beneficial for $\alpha \geq 1$, as it would only add one additional node from $M$ to her 2-neighborhood since no node is a direct neighbor of more than one node of $M$.
    Deleting a single edge is also not beneficial, as it would remove two-hop connections to at least $k \geq 3 > \alpha$ nodes, e.g.\ $r_1$ and all nodes in $L$ but $l_1$, if edge $\{m_1,r_1\}$ was deleted.
    Swapping a single edge is not beneficial either, as there is no other node that would allow $m_1$ to gain more than $k$ new two hop connections: No node has more than $k$ direct neighbors, and following the graph construction, the other edge (the one that is not swapped) would cover at least one of them already. 
    Thus, all nodes in $M$ are playing a greedy best response.
    \paragraph{Nodes in $R$:} Nodes in this set build no edges and are connected to all other nodes through incoming edges, thus they are already playing their best greedy response.
    
    Thus, the graph $G(\mathbf{s},k)$ is greedy stable for $k>2\geq \alpha \geq 1$.
\end{proof}

\noLongPaths*
\begin{proof}
        We prove the theorem by showing that there cannot be 3 (Lemma~\ref{lem:keine3deg2Knoten}) or more (Lemma~\ref{keine_k+deg2Ketten}) consecutive nodes of a degree of 2, except for the few cases described in the theorem.
\end{proof}

\keinezweidegzweiketten*
\begin{proof}
    Let $u,v_1,v_2,w\in V$ and $\{u,v_1\},\{v_1,v_2\},\{v_2,w\} \in E(\mathbf{s})$, i.e.\ the path \mbox{$u$---$v_1$---$v_2$---$w$} be a subgraph of $G(\mathbf{s})$, with $\mathrm{deg}(v_1),\mathrm{deg}(v_2)=2$, \mbox{$\mathrm{deg}(u)\neq 2$} and \mbox{$\mathrm{deg}(w)\neq2$}. Analogue, \mbox{$u'$---$v_1'$---$v_2'$---$w'$} is a similar subgraph. Note that the existence of the edges $\{v_1,v_2\}$ resp. $\{v_1',v_2'\}$ means that $u$ and $w'$ resp. $u'$ and $w'$ have to be distinct nodes.
    We prove the statement through a case analysis for different ranges of $\alpha$, considering edge directions only when necessary.
    
    For $\alpha>2$, no path of two consecutive nodes of a degree of 2 is stable, as the edge $\{v_1,v_2\}$ cannot be beneficial to its buyer.
    
    For $\alpha<2$, we distinguish the following cases:
    \begin{itemize}
    \item \emph{Case 1 $(u \neq u')$:} Due to the construction, we directly see $v_1',v_2' \notin N_2^G(v_1)$ and therefore agent $v_1$ has an improving response by connecting to either $v_1'$ or $v_2'$. (This works analogue if $w\neq w'$.)
            
    \item \emph{Case 2 $(u = u' \wedge w = w')$:} W.l.o.g.\ assume $v_2\in S_{v_1}$. Since $v_2' \notin N_2^G(v_1)$, agent $v_1$ can improve by swapping her edge $(v_1,v_2)$ to $(v_1,w)$, thereby enlarging her $2$-neighborhood by at least $v_2'$.   
    \end{itemize}
    
    For $\alpha=2$, we observe that $u$ and $w$ cannot be directly connected, as $\{v_1,v_2\}\in E(\s)$ (analogue for $u',w'$). 
    Moreover, the edges $\{u,u'\}$ resp. $\{u,w'\}$ need to exist as otherwise $v_1$ would have an incentive to connect with $v_1'$ resp. $v_2'$. Analogously, $v_2$ would have an incentive to connect with $v_1'$ resp. $v_2'$ if $\{w,u'\}\notin E(\s)$ resp. $\{w,w'\}\notin E(\s)$.
    Next, we need to analyze whether there exist other nodes that are connected to $u,w,u'$ or $w'$. For this, we introduce the notation $\Tilde{N_1}(u):=\{v\in V\setminus{\{u,v_1,v_2,w,u',v_1',v_2',w'\}} \,|\, d_G(u,v)=1\}$ representing the direct neighbors of $u$ that are not among the eight previously defined vertices.
    
    Because of symmetry, we can assume w.l.o.g.\ that $v_2\in S_{v_1}$ and $v_2'\in S_{v_1'}$.
    The first statement directly implies that $w$ has no direct neighbor in $\Tilde{N_1}(w)$ that is not part of the 2-neighborhood of $v_1$. Otherwise $v_1$ would have an incentive to swap her edge from $v_2$ to $w$.     
    In particular it means, that every node $v\in \Tilde{N_1}(w)$ needs to be a direct neighbor of $u$, i.e.\ $\Tilde{N_1}(w)\subseteq \Tilde{N_1}(u)$.
    With the same reasoning it follows from $v_2'\in S_{v_1'}$ that $\Tilde{N_1}(w')\subseteq \Tilde{N_1}(u')$.
    Moreover, we can derive $\Tilde{N_1}(u)\subseteq \Tilde{N_1}(u')$ as otherwise agent $v_1'$ would have an incentive to swap her edge $(v_1',v_2')$ to $(v_1',u)$. Analogously, we get $\Tilde{N_1}(u')\subseteq \Tilde{N_1}(u)$ as otherwise $v_1$ would benefit from swapping her edge $(v_1,v_2)$ to $(v_1,u')$.
    Thus, $\Tilde{N_1}(u')= \Tilde{N_1}(u)$ and with $\{u,u'\}\in E(\s)$ it follows directly that $\Tilde{N_1}(u')= \Tilde{N_1}(u)=\emptyset$ as otherwise there would be no incentive for establishing edge $\{u,u'\}$.
    Finally, from $\Tilde{N_1}(w)\subseteq \Tilde{N_1}(u)= \emptyset$ and $\Tilde{N_1}(w')\subseteq \Tilde{N_1}(u') = \emptyset$, we get that $\Tilde{N_1}(w)=\emptyset$ and $\Tilde{N_1}(w')=\emptyset$. This implies that the instance with $V=\{u,v_1,v_2,w,u',v_1',v_2',w'\}$ is the only possible graph structure. We have already seen an stable example in Observation~\ref*{obs:zweierketten}\ref{obs:zwei_zweierketten}.  
\end{proof}

\blaetterlemma*
\begin{proof}
    Let $u,v\in L$ be two leaves connecting to two different nodes $x,y\in V$. Since $G(\mathbf{s})$ is stable and $\alpha>1$ holds, agents~$u$ and $v$ buy their respective incident edge $(u,x)$ and $(v,y)$. As these edges connect to different nodes, $u$ and $v$ are not contained in each other's 2-neighborhood. Let w.l.o.g.\ $N_2^G(u) \geq N_2^G(v)$, then $S_v'=\{x\}$ would be an improving response to her former strategy $S_v=\{y\}$, as agent~$v$ gains at least one node ($u$) more by this swap, as it would lose. 
\end{proof}

\lemNEdiamthreeonealphathree*
\begin{proof}[Proof (continued)]
        It remains to show that none of the nodes of the sets $A$, $B$ or $V_1$ can improve their strategies. Because of symmetry in the construction and the regularity of the edges, considering one agent of each set is enough, and agents for nodes in $V_2,\ldots,V_5$ behave analogous to agents in $V_1$.
    \paragraph{Nodes in $A$:} For $1\leq i \leq 5$, let $a_i$ be an arbitrary node of $A$. 
    Agent $a_i$ only builds one edge to $a_{i+1\bmod 5}$ and has all nodes except $b_{1+3i\bmod5}$ in her 2-neighborhood. Her edge to $a_{i+1\bmod 5}$ is ensuring two-hop connections to $a_{i+1\bmod 5}$ and especially to all nodes in $V_{1+3i\bmod5}$. Deleting that edge is not beneficial, as $a_i$ would lose two-hop connections to $\lceil\alpha\rceil+1$ nodes.
   
    Thus, adding an edge does not result in improvement, as at most one node can be gained.
    Swapping an edge is not beneficial either, as there is no node that has $a_{i+1\bmod 5}$, $V_{1+3i\bmod5}$, and $b_{1+3i\bmod5}$ in their $1$-neighborhood.
    Therefore, all $a_i$ are stable.
    \paragraph{Nodes in $B$:} 
    The key insight is that the graph $G(\s)$ exhibits a symmetry arising from the way the two 5-cycles are connected through the nodes of $V_i$.
    This means if we match each node $a_i$ with a specific node $b_{1+3i\bmod5}$, e.g.\ $a_1$ and $b_4$, we can see by going over all connections, that their 2-neighborhoods are identical (only missing each other) and their possible moves can be mapped to each other.
    Since we have shown before that all $a_i$ are stable, it transfers directly that all corresponding $b_j$ are stable as well.
   
    \paragraph{Nodes in $V_1$:} All $\lceil\alpha\rceil$ agents in $V_1$ have the same strategy, so analyzing only one of them is sufficient. 
    Agent $v\in V_1$ builds four edges, to $a_{2i-1\bmod5}=a_1$, $a_{2i+1\bmod5}=a_3$, $b_{i\bmod5}=b_1$, and $b_{i+2\bmod5}=b_3$, and has no incoming edges.
    With any of these edges, $v$ gains all other nodes in $V_1$ into her 2-neighborhood. Moreover, each edge adds exclusive connections to two distinct nodes from $A$ or $B$ and one other  group $V_i$ of size $\lceil\alpha\rceil$.
    The two edges of the 5-cycle in $A$ resp. $B$ offer a common connection to $a_2$ resp. $b_2$.
    With these four edges, $v$ gets all other nodes into her 2-neighborhood, thus adding or swapping a single edge is not beneficial.

    Deleting a single edge is not beneficial either, as it would remove at least $\lceil\alpha\rceil+2$ nodes -- two on a 5-cycle and one other group $V_i$.
    Thus, the only improving response could be to delete at least one edge and swap one or more other edges.
    Following from the construction and the common endpoints of the edges from different $V_i$ nodes, at least three edges are needed to get all $V_i$ nodes into her 2-neighborhood. 
    As each $V_i$ contains more than $\alpha$ nodes, losing one of these groups would be unacceptable.
    With three edges, e.g.\ to $a_3$, $b_1$, $b_5$, all nodes but three can be covered, which is the best possible due to the way the different $V_i$ connect to both cycles. 
    For $\alpha\leq3$, this multi-swap move is non-improving, meaning $v$ and thus by symmetry all nodes in all $V_i$ are Nash stable. 
    Since the only improving moves for $\alpha > 3$ are moves where multiple edge changes are involved, the network remains greedy stable for higher $\alpha$.  
\end{proof}

\GEdiamFour*
\begin{proof}
We prove the theorem by providing an upper bound in Lemma~\ref{lem:no-GE-alpha>=5} and a lower bound example in Lemma~\ref{lem:GE_diam4_a>=3}.
\end{proof}

\GEdiamThree*
\begin{proof} A lower bound instance has already been provided in Lemma~\ref{lem-NE-diam3-1alpha3}.
For showing the upper bound, let $G(\mathbf{s})=(V,E(\mathbf{s}))$ be a stable network for a fixed $\alpha<3$. Assume $\mathrm{diam}(G(\mathbf{s}))=4$ holds, thus there is a path 
\mbox{$v_1$---$v_2$---$v_3$---$v_4$---$v_5$}
with \mbox{$d_G(v_1,v_5)=4$}. For $\alpha=1$, $v_1$ has an improving response in connecting to one of the two, since $v_4,v_5\notin N_2^G(v_1)$. It remains to prove the statement for $\alpha>1$.
As shown in Lemma~\ref{lem-no-two-leaves-at-different-nodes}, $v_1$ and $v_5$ cannot both be leaves, and thus w.l.o.g.\ $\mathrm{deg}(v_1)\geq 2$ holds. Since $d_G(v_1,v_5)=4$, we know that for all direct neighbors $w\in N_1^G(v_1)$, $d_G(w,v_5)\geq3$ holds, too.
Thus, $v_5$ can improve by building an edge to $v_1$, as this adds at least $3>\alpha$ nodes to her 2-neighborhood.
This shows for $1 \leq \alpha<3$, a stable network $G(\s)$ can have at most a diameter of 3.  
\end{proof}

\dynamics*

\begin{proof}
We start by describing a game with a sequence of strategy profiles $\mathcal{C} = \s_0,\s_1,\s_2,\s_3,\s_4,\s_5,\s_0$ (see Fig.~\ref{fig:IRC-Example}) and then show that this sequence indeed forms an improving response cycle. The game consists of the following agents: 
\begin{itemize}
    \item nodes $a_1$, $a_2$, and $b$, whose improving responses will form the cycle,
    \item nodes $x_i$ and $y_i$, for $i\in \{1,2\}$, between which $a_i$ will swap her edge,
    \item three sets of nodes $V_{a}^i, V_{y}^i, V_{xb}^i$, which will incentivize $a_i$'s and $b$'s strategy changes. The sizes of these sets are $|V_{a}^i|=1$, $|V_{y}^i|=2$ and \mbox{$|V_{xb}^i| = 4$}.
\end{itemize}
The cycle $\mathcal{C}$ starts with $\s_0$, which is constructed as follows (see Fig.~\ref{fig:IRC-Example}(a)). Agents $a_i$ builds an edge to the respective node $x_i$; agent~$b$ builds to node~$a_1$; the agents~$x_i$ build no edges; the agents $y_i$ build to $b$; all nodes in $V_{xb}^i$ build to the respective node~$x_i$ and to node~$b$; all nodes in $V_{a}^i$ build to the respective node~$a_i$; and all nodes in $V_{y}^i$ build to the respective node~$y_i$.

The following changes happen between the strategy profiles of $\mathcal{C}$, all other strategies stay the same:
\begin{itemize}[leftmargin=1.5cm]
    \item[$\s_0\to\s_1$:] $S_{a_1}'=\left(S_{a_1}\setminus\{x_1\}\right)\cup\{y_1\}=\{y_1\}$, i.e., $a_1$ swaps her edge \mbox{from $x_1$ to $y_1$}, 
    \item[$\s_1\to\s_2$:] $S_b'\phantom{_1}=\left(S_{b}\setminus\{a_1\}\right)\cup\{a_2\}=\{a_2\}$, i.e., $b$ swaps her edge \mbox{from $a_1$ to $a_2$,}
    \item[$\s_2\to\s_3$:] $S_{a_2}'=\left(S_{a_2}\setminus\{x_2\}\right)\cup\{y_2\}=\{y_2\}$, i.e., $a_2$ swaps her edge \mbox{from $x_2$ to $y_2$,} 
    \item[$\s_3\to\s_4$:] $S_{a_1}'=\left(S_{a_1}\setminus\{y_1\}\right)\cup\{x_1\}=\{x_1\}$, i.e., $a_1$ swaps back from $y_1$ to $x_1$, 
    \item[$\s_4\to\s_5$:] $S_b'\phantom{_1}=\left(S_{b}\setminus\{a_2\}\right)\cup\{a_1\}=\{a_1\}$, i.e., $b$ swaps back from $a_2$ to $a_1$,
    \item[$\s_5\to\s_0$:] $S_{a_2}'=\left(S_{a_2}\setminus\{y_2\}\right)\cup\{x_2\}=\{x_2\}$, i.e., $a_2$ swaps back from $y_2$ to $x_2$.
\end{itemize}
Fig.~\ref{fig:IRC-Example} depicts the improving response cycle that holds for any edge price $\alpha$, since only edge-swaps are performed. 
Dashed edges indicate edges from $a_1$, $a_2$, and $b$, i.e., the changing edges.
We will now discuss all six strategy changes and show there that each is an improving response for the respective node for any $\alpha$. 

{$\s_0\to\s_1:$} Agent $a_1$ does not need her edge $(a_1,x_1)$ to get all nodes in $V_{xb}^1$ into her 2-neighborhood. 
Since the two nodes in $V_{y}^1$ are not in her 2-neighborhood, she swaps her edge from $x_1$ to $y_1$.
This is beneficial, since agent~$a_1$ adds the two nodes in $V_{y}^1$ to her 2-neighborhood while keeping her two-hop connection to the nodes in $V_{xb}^1$. She only loses node~$x_1$ from her 2-neighborhood. Since the number of built edges does not change, overall her cost decreases by 1. 

{$\s_1\to\s_2:$} Agent $b$ has all nodes except $a_2$ and all nodes in $V_{a}^2$ in her 2-neighborhood.
Via the new edge $(a_1,y_1)$, agent $b$ gets another way to connect to $a_1$, i.e.\ in two hops over $y_1$. 
Thus, swapping her edge to $a_2$ is beneficial, as agent~$b$ adds $|V_{a}^2\cup\{a_2\}| = 2$ nodes to her 2-neighborhood, while she only loses the single node in $V_{a}^1$, as $a_1$ is still reachable in two hops via $y_1$. Overall, since she still has the same number of edges, her cost decreases by 1.

{$\s_2\to\s_3:$} It is the same as $\s_0\to\s_1$, but from agent $a_2$'s perspective.

{$\s_3\to\s_4:$} After the deletion of edge $(b,a_1)$, agent $a_1$ loses her two-hop connections to all nodes in $V_{xb}^1$.
Together with $x_1$, these are one node more than $|V_{y}^1\cup\{b,y_1\}|$. 
Thus, a swap back to $x_1$ is beneficial for agent~$a_1$.

{$\s_4\to\s_5:$} It is the same as $\s_1\to\s_2$, but with the roles of $a_1$ and $a_2$ swapped.

{$\s_5\to\s_6=\s_0:$} It is the same as $\s_3\to\s_4$, but from agent $a_2$'s perspective.

The resulting state is the same as $\s_0$, as over the six strategy changes three edges were swapped and swapped back.
Thus, the sequence $\mathcal{C}$ is indeed an improving response cycle.
\end{proof}

\hardness*

\begin{proof}
The formal definition of the decision problem is as follows: 
Given a graph $G(\mathbf{s})=(V,E(\mathbf{s}))$, an isolated node $x\in V$, an edge cost $\alpha\in \mathbb{R}_{>0}$, and a maximum strategy size $k\in\mathbb{N}$. The objective is to decide whether it is possible for $x$ to choose a strategy $S_x$ with $|S_x|\leq k$, so that $cost_x((S_x, \s_{-x}))\leq \alpha \cdot k + (n - k) \lfloor\alpha\rfloor$.

The membership to NP can be shown with a polynomial time verification \cite{AroraBarak2009} by encoding a strategy $S_x$ in a bit-string certificate $c$ of length $n$. Hereby, $c[i]=1$ means $x$ builds an edge to $v_i \in V$, 0 means $x$ does not.
Given an instance $(G,x,\alpha,k)$ and a certificate $c$, the verifier builds all edges encoded in $c$ and calculates $cost_x((S_x,\s_{-x}))$. 
For that, the strategies of $\s_{-x}$ can be deduced from $G$ in $\bigO(n+m)$ by going over all edges. The cost of $x$ can be calculated by using a Breadth-First Search from $x$ and collecting all neighbors that are not in its 2-neighborhood. Finally, the verifier accepts the instance if and only if $cost_x((S_x, \s_{-x}))\leq \alpha \cdot k + (n - k) \lfloor\alpha\rfloor$ holds.

If $(G,x,\alpha,k)$ is a "Yes"-instance, there is a feasible strategy that can be encoded as a certificate, so that the verifier accepts it.
If the verifier accepts an instance with a certificate, the certificate encodes a feasible strategy so that $cost_x$ fulfills the requirement and the instance is indeed a "Yes"-instance.
Together, we can conclude the decision problem belongs to NP.

The NP-hardness reduction we construct as follows. Given a DS instance $(G,k)$, where $G$ is a graph and $k$ the size of a dominating set. We construct the state $\s$ by describing the corresponding graph $G(\mathbf{s})$. The orientation of the edges is arbitrary. First, we take a copy of $G$ including all nodes and edges. For each node $v \in V$, we add $\lfloor\alpha\rfloor$ copies, i.e., nodes $v_i$ with edges $(v_i,v)$ for $i \in \left\{1, \dots, v_{\lfloor\alpha\rfloor}\right\}$. Finally, we add an isolated node $x$ for which we will calculate the best response.
For the edge budget, we transfer $k$ directly.
It remains to prove correctness of the construction.

First, we show that given a DS of size $k$ in $G$, there is a best response for $x$ to $G(\mathbf{s})$ with $cost_x((S_x, \s_{-x}))\leq \alpha \cdot k + (n - k) \lfloor\alpha\rfloor$. Let $D$ be the DS with $|D|=k$ for the original graph $G$. In our transformed instance, buying all $k$ edges from $x$ to $v\in D$ brings all original nodes from $G$ as well as all the copies of the $k$ chosen nodes into $N_2^{G(\s)}(x)$. Thus, there is a strategy $S_x$ with $cost_x(\s) = k \alpha + (n - k) \lfloor\alpha\rfloor$ and the best response is at most as costly.

Second, suppose there is a best response $S_x$ with $cost_x((S_x,\s_{-x}))\leq \alpha\cdot k + (n - k) \lfloor\alpha\rfloor$ in 2-NMG. We show that this implies a DS of size $k$ in $G$. As strategy~$S_x$ is a best response we can assume that $S_x$ does not contain any of the copied vertices as they have degree 1 and swapping the edge from a copy $v_i$ to the original node $v$ would never increase cost. 
Furthermore, if there is an original agent $v\notin N_2^{G(\s)}(x)$, all her copies are also not in $N_2^{G(\s)}(x)$. 
In this case agent~$x$ could add an edge with cost $\alpha$ to $v$ and increase $\left|N_2^{G(\s)}(x)\right|$ by $\lfloor \alpha \rfloor + 1 > \alpha$. 
Because of $|S_x| \leq k$, it holds that $cost_x((S_x,\s_{-x})) = \alpha \left|S_x\right| + \left(n-|S_x|\right) \lfloor\alpha\rfloor \leq \alpha\cdot k + (n - k) \lfloor\alpha\rfloor$. Thus, agent~$x$ has a best response with at most $k$ nodes, all nodes are original nodes, and all original nodes are in $N_2^{G(\s)}(x)$. Thus, $S_x$ forms a dominating set with $|S_x|\leq k$ in $G$. 
\end{proof}

\begin{restatable}{lemma}{keinethreedegtwoKnoten}
\label{lem:keine3deg2Knoten}
    Let $G(\mathbf{s})=(V,E(\mathbf{s}))$ be a stable network for any $\alpha$. Then there exists no path with exact length 3 as a subgraph, i.e.
    \begin{equation*}
    \begin{split}
        \nexists\, u, v_1,v_2,v_3, w \in V: \mathrm{deg}(v_{j\in[1,3]})=2\, \wedge \mathrm{deg}(u),\mathrm{deg}(w)\neq 2\, \\ \wedge \,d_G(u,v_1)=d_G(v_1,v_2)=d_G(v_2,v_3)=d_G(v_3,w)=1
    \end{split}
    \end{equation*}
    or $\alpha=1$ and $G(\mathbf{s})$ is a $4$-cycle with a leaf.
\end{restatable}

\begin{proof}
    Let $v_1,v_2,v_3 \in V$, $ d_G(v_1,v_2)=d_G(v_2,v_3)=1$ and $\mathrm{deg}(v_1)=\mathrm{deg}(v_2)=\mathrm{deg}(v_3)=2$. 
    For $u,w\in V$ consider the subgraph \mbox{$u$---$v_1$---$v_2$---$v_3$---$w$} with $\mathrm{deg}(u), \mathrm{deg}(w)\neq 2$ (for longer paths see Lemma~\ref{keine_k+deg2Ketten}). First, we consider the case $\alpha<2$.
    \begin{itemize}
    \item \emph{Case 1 $(u=w)$:} The nodes $u,v_1,v_2,v_3 \in V$ form a $4$-cycle and we have to differentiate cases depending on the direction of the edges adjacent to agent $v_2$. 
    If $v_2$ has at least one outgoing edge, she can improve her strategy by swapping her edge from either $v_1$ or $v_3$ to $u$. 
    Because of $\mathrm{deg}(u)\geq 3$ she would add at least one more node into her neighborhood while still keeping all the other agents. If $v_2$ does not have any outgoing edges and $1<\alpha<2$, then either $u$, $v_1$ or $v_3$ is purchasing two edges and can improve by deleting one. If $\alpha<1$, then $v_2$ can improve her strategy since there exists at least one neighbor $u'$ of $u$ with $d_G(v_2,u')=3$. For $\alpha=1$, there exists a particular $4$-cycle with a leaf that is stable (see Observation~\ref*{obs:NEs}\ref{obs:4cycleleaf}).
    \item \emph{Case 2.1 ($d_G(u,w)=1)$:} In case either $\mathrm{deg}(u)\geq 4$ or $\mathrm{deg}(w)\geq 4$, agent $v_2$ can improve her strategy by connecting to this agent and thus gain at least two new neighbors into her 2-neighborhood. 
    Because of $\alpha<2$ this is a beneficial deviation. 
    Whereas if $\mathrm{deg}(u)=\mathrm{deg}(w)=3$ and $u',w'\in V$ with $u'\neq w'$ and $d_G(u',u)=d_G(w',w)=1$, both agents $v_2,v_3 \notin N_2^G(u')$ and thus agent $u'$ has an improving response by connecting to $v_2$. 
    The case where $u'=w'$ can be neglected for $1<\alpha<2$ since the connection between $u$ and $w$ would not be beneficial. 
    For $\alpha\leq1$, we have to distinguish two cases: \begin{itemize}
        \item 2.1.1: If $u$ purchases an edge to $u'$, it implies that $u'$ has at least one neighbor $u''$ that $u$ gains into her 2-neighborhood with that edge. For this agent $u''$ it holds by construction that $d_G(u'',v_1)=d_G(u'',v_3)=3$ and $d_G(u'',v_2)=4$. Thus, the network is not stable.
        \item 2.1.2: If $u'$ purchases an edge either to $u$ or $w$, she can improve her strategy by swapping her edge one node further and gaining additionally $v_2$ into her 2-neighborhood, as the original node ($u$ or $w$) has no neighbors $u'$ would lose, since $\mathrm{deg}(u)=\mathrm{deg}(w)=3$.
    \end{itemize}
    \item \emph{Case 2.2 $(d_G(u,w)\geq 2)$:} If w.l.o.g.\ $\mathrm{deg}(u)\geq 3$, agent $u$ has at least two neighbors that are not part of $N_2^G(v_2)$, and therefore agent $v_2$ can improve her strategy by connecting to $u$. Whereas if w.l.o.g.\ $\mathrm{deg}(u)=1$, it is straightforward to see that $v_3,w \notin N_2^G(u)$ and therefore an additional edge is beneficial for $u$. 
    \end{itemize}
    If $\alpha> 2$ it is easy to see that there always exist at least one edge between two nodes with a degree of two that can be deleted since it increases the respective 2-neighborhood by only 2, e.g.\ edge $(v_1,v_2)$.
    
    For $\alpha =2$ we observe that if $v_2$ constructs the edge to $v_1$, it is easy to see that swapping $(v_2,v_1)$ to $(v_2,u)$ increases $v_2$'s 2-neighborhood, if $\mathrm{deg}(u)\geq 3$. Thus, $v_1$ constructs the edge $(v_1,v_2)$. Analogously, $v_3$ constructs the edge $(v_3,v_2)$. We distinguish the following cases.  The instance where $\mathrm{deg}(u)=1$ is dealt with in the last case of this proof. 
    \begin{itemize}
        \item \emph{Case 1 ($u=w$):} Regardless of the direction, deleting the edge $\{v_1,v_2\}$ is an improving response, since it decreases the respective 2-neighborhoods by only 1 node.
        \item \emph{Case 2 ($d_G(u,w) = 1$):} If $\max\{\mathrm{deg}(u), \mathrm{deg}(w)\} \geq 4$, we assume w.l.o.g.\  \mbox{$\mathrm{deg}(w) \geq 4$}. Swapping the edge $(v_1,v_2)$ to $(v_1,w)$ increases $|N_2^G(v_1)|$ by one. Thus, $\mathrm{deg}(u)= \mathrm{deg}(w)= 3$. Let $u'$ and $w'$ denote the neighbors of $u$ and $w$ that are not in the \mbox{cycle $u$---$v_1$---$v_2$---$v_3$---$w$---$u$}, respectively.
        \begin{itemize}
            \item \emph{Case 2.1 (there are no further nodes, i.e.\ $V=\{u,v_1,v_2,v_3,w,u',w'\}$):} Either $u'$ and $w'$ are both leaves which is an immediate contradiction to Lemma~\ref{lem-no-two-leaves-at-different-nodes}. Otherwise, there is an edge $\{u',w'\}$. Observe that buying this edge is not beneficial both for $u'$ and $w'$. 
            \item \emph{Case 2.2 (there is an additional node $v$):} By construction $v$ cannot have edges to $u, v_1, v_2, v_3$, and $w$. Thus, $v_1,v_2,v_3 \notin N_2^G(v)$, therefore $v$ has an improving response by buying to $v_2$.
        \end{itemize}
        \item \emph{Case 3 ($d_G(u,w) \geq 2$):} If $\mathrm{deg}(u),\mathrm{deg}(w)\geq3$, the agent who purchased the edge $(v_1,v_2)$ can improve her strategy by swapping her connection one node further. If w.l.o.g.\ $\mathrm{deg}(u)=1$, then $u$ can improve her strategy by swapping her edge from agent $v_1$ to agent $v_2$ gaining $v_3$ into her 2-neighborhood. 
    \end{itemize}
\end{proof}

\begin{restatable}{lemma}{keinekpluszweierketten}

\label{keine_k+deg2Ketten}
    Let $G(\mathbf{s})=(V,E(\mathbf{s}))$ be a stable network for any $\alpha$. Then, either
    \begin{equation*}
        \begin{split}
        \nexists\, v_1,\dots, v_k\in V, k\in \mathbb{N}_{>3}: \mathrm{deg}(v_{j\in[1,k]})=2
         \wedge \forall i \in [1,k-1]\colon d_G(v_i,v_{i+1})=1
        \end{split}
    \end{equation*}
    or  $\alpha =1$ and $G(\s)$ is a 4-cycle, $\alpha =2$ and $G(\s)$ is a 5-cycle with an additional leaf or $\alpha\leq 2$ and $G(\mathbf{s})$ is a $5$-cycle.
\end{restatable}

\begin{proof}
    Let  $v_1,\dots, v_k\in V, k\in \mathbb{N}_{>3}$ with $\mathrm{deg}(v_{j\in[1,k]})=2
         $ and \mbox{$ \, \forall i \in [1,k-1] :$} $d_G(v_i,v_{i+1})=1 $, i.e.\ $G$ contains a chain of $k$ many nodes of a degree of 2.
         Further, consider the subgraph
         \mbox{$u$---$v_1$---$v_2$---$v_3$---$\dots$---$v_k$---$w$} with $u,w\in V$ and $\mathrm{deg}(u),\mathrm{deg}(w)\neq 2$.
         Without loss of generality we can assume $v_3\in S_{v_2}$. First we analyse the case $\mathrm{deg}(u),\mathrm{deg}(w)\geq 3$ and start with $\alpha<2$.
         \begin{itemize}
         \item \emph{Case 1 $(u=w)$:} Because of $\mathrm{deg}(u)\neq 2$ we know that at least one node $u'\in V$ exists with $d_G(u,u')=1$ and $u'\notin \{v_1,\dots, v_k\}$. It follows from the construction that $v_2,v_3 \notin N_2^G(u')$ and therefore agent $u'$ can improve her strategy by connecting to $v_2$ or $v_3$. 
         
         \item \emph{Case 2 $(u\neq w)$:} It is clear that $w \notin N_2^G(v_2)$ and thus agent $v_2$ can improve her strategy by deleting her edge to $v_3$ and connecting instead with $w$. 
         \end{itemize} 
         If $\alpha> 2$, analogously to the proof of Lemma~\ref{lem:keine3deg2Knoten} it is easy to see that there always exist at least one edge between two nodes with a degree of two that can be deleted.
         
         For $\alpha=2$ we distinguish the cases as above.
         \begin{itemize}
         \item \emph{Case 1 $(u=w)$:} Either $\mathrm{deg}(u)\geq 4$, then the player constructing edge $\{v_2,v_3\}$ would swap the edge to $u$. 
         If $\mathrm{deg}(u)=3$, thus let $u$---$v_1$---$v_2$---$v_3$---$\dots$---$v_k$---$u$ be at least a 6-cycle with an additional connection at agent $u$. Let w.l.o.g.\ $v_2$ construct the edge $(v_2,v_3)$. Swapping to $(v_2,v_4)$ strictly decreases $v_2$'s cost. The only remaining case is the 5-cycle with a single additional leaf. This is exactly the network analyzed in Observation~\ref*{obs:NEs}\ref{obs:5cycleleaf} and stable for $\alpha = 2$.
         
         \item \emph{Case 2 $(u\neq w)$:} It is clear that $w \notin N_2^G(v_{k-2})$ and thus agent $v_{k-2}$ can improve her strategy by deleting her edge to $v_{k-1}$ and connecting instead with $v_k$. This is an improving response since she gains $w$ into her 2-neighborhood.
         \end{itemize}
         If w.l.o.g $\mathrm{deg}(u)=1$, it is straightforward to see that agent $u$ can improve her strategy by either shifting her edge from agent $v_2$ one hop further to agent $v_3$ for $\alpha\leq 2$ or by removing her edge $(v_1,v_2)$ for $\alpha>2$.
         
         Finally, we have to analyse the special case of $u=w$ and $\mathrm{deg}(u)=2$, namely cycles. For $\alpha\leq 1$, we know that a $4$-cycle is stable (see Observation~\ref*{obs:NEs}\ref{obs:4cycle}) and for $\alpha\leq2$, the $5$-cycle is stable (see Observation~\ref*{obs:NEs}\ref{obs:5cycle}). For larger $\alpha$, these networks become instable, as the benefit of an edge in a cycle cannot exceed $2$. Moreover, the $4$-cycle is the smallest cycle that can be stabilized since in a triangle there is no incentive for the third edge. We will now show that bigger cycles are never stable. Let $G=(V,E(\s))$ be a $l$-cycle with $l\geq 6$ and $u,w\in V$ and w.l.o.g $(u,w)\in E(\s)$. The agent that purchases this edge can improve her strategy by swapping the edge one node further in the cycle increasing her neighborhood by one agent.   
\end{proof}

\begin{lemma}\label{lem:no-GE-alpha>=5}
    All connected greedy stable networks $G(\mathbf{s})$ fulfill $\mathrm{diam}(G(\mathbf{s}))\leq4$.

\end{lemma}
\begin{proof}
    Let $G(\mathbf{s})=(V,E(\mathbf{s}))$ be a stable network with diameter at least 5. 
    Then there are $v_1,v_6\in V$ with $d_G(v_1,v_6)=5$. 
    This means $N_2^G(v_1) \cap N_2^G(v_6)=\emptyset$. 
    Also, it exists $v_2$ with $d_G(v_2,v_6)=4$ and $d_G(v_2,v_1)=1$ and an edge $e=\{v_1,v_2\}$. Let $x$ denote the buyer of $e$, and $y$ the node the edge is built to. Since $G$ is stable, building $e$ increases the size of $N_2^G(x)$ by at least $\lceil\alpha\rceil$ nodes. All these nodes are contained in $N_1^G(y)$. Together with $x$ we have $N_1^G(y)\geq 1+ \lceil\alpha\rceil$. Since $d_G(v_6,y)\geq 4$, $N_1^G(y)$ and $N_2^G(v_6)$ are disjoint. Thus $v_6$ has an improving response by buying the edge to $y$ as well.   
\end{proof}

\noindent To prove the tightness of our bound, we show the following lemma.

\begin{restatable}{lemma}{GEdiamfour}
\label{lem:GE_diam4_a>=3}
   
   There is a greedy stable network $G$ with $\mathrm{diam(G)}=4$. 
\end{restatable}

\begin{proof}
\begin{figure}[t]
\centering
\includegraphics[scale=0.4]{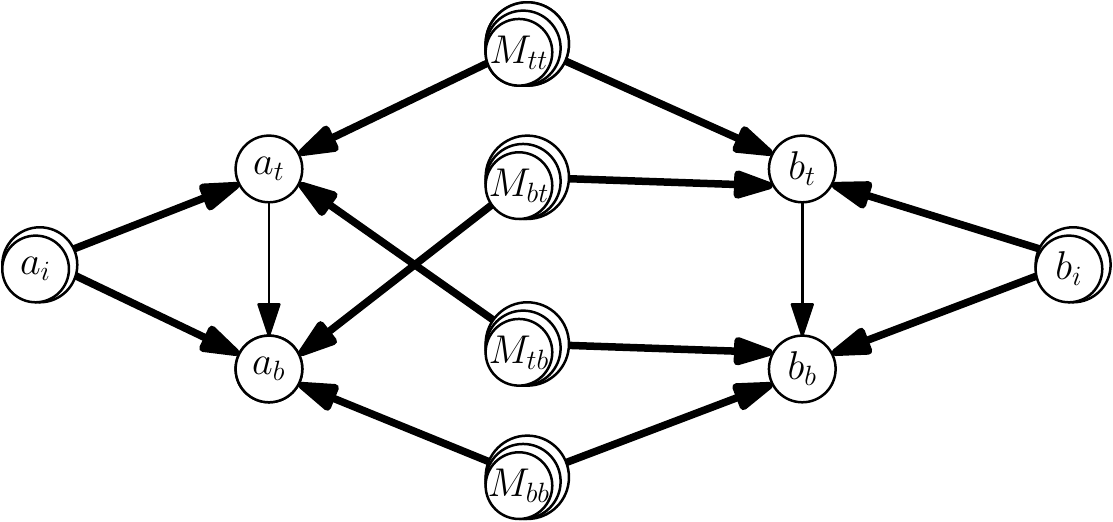}
\caption{Illustration of the construction for GE with a diameter of 4 (Lemma~\ref{lem:GE_diam4_a>=3}). Nodes with multiple borders and bold edges represent multiple identical nodes.}\label{fig-GE-Diam4-Example}
\end{figure}

We start by explaining how the networks are constructed, followed by a proof of greedy stability.
For a given $\alpha\geq3$ the network consists of the following nodes ($t$ for top, $b$ for bottom, see Fig.~\ref{fig-GE-Diam4-Example}):
\begin{itemize}
    \item Set $A$ containing $a_t$, $a_b$, and $a_i$ with $i\in \mathbb{N}$ and $1\leq i \leq \lfloor\alpha\rfloor-2$
    \item Set $B$ analogue to $A$
    \item Four Sets $M_{tt}$ $M_{tb}$ $M_{bt}$ $M_{bb}$ of exactly $\lfloor\alpha\rfloor$ nodes each
\end{itemize}
These nodes are connected through the following strategies:
\begin{itemize}
    \item $S_{a_t}=\{a_b\}$ and $S_{b_t}=\{b_b\}$
    \item $S_{a_i}=\{a_t,a_b\} $ and $S_{b_i}=\{b_t,b_b\}$ for all $i\in \mathbb{N}$ and $1 \leq i \leq \lfloor\alpha\rfloor-2$
    \item $S_{v}=\{a_x,b_y\}$ for every node $v\in M_{xy}$ and $x,y \in \{t,b\}$
\end{itemize}
This network has diameter 4. 

    To prove stability of the construction, it is sufficient to analyze the possibility of improving responses only for nodes of the set $A$ and $M_{tt}$ since nodes in $B$ resp. $M_{tb}$, $M_{bt}$, or $M_{bb}$ behave in a analogue way. 
    
    \paragraph{Node $a_t$:} This node builds one edge to $a_b$ and has all nodes except all $b_i$ in her 2-neighborhood. 
    Her edge to $a_b$ is needed for two-hop connections to the nodes of $M_{bt}$ and $M_{bb}$. 
    Connections to all nodes in $A$, $M_{tt}$, $M_{tb}$, and $b_t$, $b_b$ are established via incoming edges. 
    All nodes $b_i$ are at distance 3.
    
    This means establishing an additional edge is not beneficial, as at most $\lfloor\alpha\rfloor-2 < \alpha$ nodes are gained into her 2-neighborhood, i.e.\ all nodes $b_i$.
    Deleting the edge to $a_b$ is also not profitable, as $a_t$ would lose all two-hop connections to all nodes of $M_{bt}$ and $M_{bb}$, which are $2\lfloor\alpha\rfloor>\alpha$ many.
    Swapping an edge is no improving response either, as $a_t$ would again lose all nodes of $M_{bt}$ and $M_{bb}$, and there is no node that has more than these nodes in her direct neighborhood, i.e.\ nodes from $M_{bt}$, $M_{bb}$ and any $b_i$. 
    Thus $a_t$ has no incentive to deviate from her strategy.
    \paragraph{Node $a_b$:} This node builds no edges and has every node except all $b_i$ in her 2-neighborhood. 
    Additional edges cannot be improving responses, as at most all $b_i$ nodes ($\lfloor\alpha\rfloor-2 < \alpha$) are gained.
    Therefore $a_b$ has no incentive to deviate from her strategy, too.
    \paragraph{Nodes $a_i$:} All nodes $a_i$ follow the same strategy, so we analyze a representative $a_1$.
    This node builds two edges (to $a_t$ and $a_b$) and her 2-neighborhood contains all nodes of $A$ and of $M_{xy}$. The nodes of $B$ are at distance 3 or 4.
    Each of her edges allows two-hop connections to all other $a_i$, $a_t$, and $a_b$. The edge to $a_t$ exclusively adds the nodes of $M_{tt}$ and $M_{tb}$ to $N_2^G(a_1)$, and the edge to $a_b$ exclusively adds $M_{bt}$ and $M_{bb}$ to $N_2^G(a_1)$.
    Adding an edge is not an improvement for $a_1$, as only $|B|=\lfloor \alpha \rfloor$ nodes are not in $N_2^G(a_1)$.
    Deleting or swapping one edge is not beneficial for $a_1$, as this would imply analogous changes to her 2-neighborhood as previously discussed for $a_t$ and her edge to $a_b$, due to the symmetrical nature of the network.
    Thus all nodes $a_i$ are stable. As this covers all nodes in $A$, all nodes in $B$ are stable as well.
    \paragraph{Nodes of $M_{tt}$:} All nodes in $M_{tt}$ behave identical, so we analyze one $m\in M_{tt}$ only. 
    This node builds two edges (to $a_t$ and $b_t$), her 2-neighborhood contains all nodes in $A$, $B$, $M_{tt}$, $M_{tb}$, and $M_{bt}$. 
    All nodes in $M_{bb}$, i.e.\ nodes that build to neither $a_t$ nor $b_t$, are at distance 3.
    Either of her edges allows two-hop connections to all other nodes in $M_{tt}$.
    Only the edge to $a_t$ adds all nodes in $A$ and $M_{tb}$ to $N_2^G(m)$, and the edge to $b_t$ does the same for nodes in $B$ and $M_{bt}$.
    Adding an edge is not an improvement for $m$, as only $|M_{bb}|=\lfloor \alpha \rfloor$ nodes are not in $N_2^G(m)$.
    Deleting one edge is not beneficial, as $m$ would lose more than $\alpha$ nodes from her 2-neighborhood, either all nodes in $A$ and $M_{tb}$, or all nodes in $B$ and $M_{bt}$.
    Swapping one edge is not beneficial, as there is no way to get more than $|A|+|M_{tb}|= 2 \lfloor\alpha\rfloor$ nodes with one edge, while keeping the other edge where it is.
    Thus, all nodes in $M_{tt}$ are stable, and by symmetry all other $M_{xy}$ sets.
    This shows that this network is greedy stable for $\alpha\geq3$.  
\end{proof}
\end{document}